\documentclass[reqno]{amsart}

\usepackage{style}

\author{Mireille Boutin}
\address{Department of Mathematics, Purdue
  University, 150 N.~University St., West Lafayette, IN, USA 47907}
\email{mboutin@purdue.edu}
\author{Gregor Kemper} \address{Technische Universit\"at M\"unchen,
  Department of Mathematics, Boltzmannstr. 3, 85748 Garching, Germany}
\email{kemper@tum.de}

\title[Global Positioning: Uniqueness and Solution]{Global Positioning: the Uniqueness Question and a New Solution Method}

\date{October, 2023}


\begin{document}


\begin{abstract}
We provide a new algebraic solution procedure for the global positioning problem in $n$ dimensions using $m$ satellites. We also give a geometric characterization of the situations in which the problem does not have a unique solution. This characterization shows that such cases can happen in any dimension and with any number of satellites, leading to counterexamples to some open conjectures. We fill a gap in the literature by giving a proof for the long-held belief that when $m \ge n+2$, the solution is unique for almost all user positions. Even better, when $m \ge 2n+2$, almost all satellite configurations will guarantee a unique solution for {\em all} user positions. 

Some of our results are obtained using tools from algebraic geometry. 
\end{abstract}

\maketitle

\section*{Introduction} \label{sIntro}%


Positioning systems play an important role in many applications. A commonly-used positioning system is the GPS (originally Navstar GPS), a satellite-based system owned by the United States, which allows users to determine their location. It essentially works as follows. The satellites are equipped with a precise clock and know their position at all times. They broadcast this information (position and time) continuously using radio waves. A user equipped with a radio receiver is assumed to receive the data from at least four satellites. Using the information received, the user computes their own position in 3D. The clocks of the satellites are assumed to be precisely synchronized together, but not with the clock of the receiver.  So the difference between the time-of-arrival on the local clock of the receiver  
and the time of emission as recorded by the clock on a satellite is equal to the time taken by the signal to travel from the satellite to the receiver (time-of-flight) plus an unknown offset which is the same for all satellites. Since radio waves travel at a known constant speed (the speed of light), the time-of-flight determines the distance between the satellite position and the user position (the range). Therefore the receiver obtains the {\em pseudo-range} values, which are the values of the ranges shifted by the unknown offset.

In general, the global positioning problem consists in determining the location of a receiver given the location of several sources and the distance to each source, up to an unknown offset. This is equivalent to determining the position of the source of a signal emitted at an unknown time and received by several receivers at known positions equipped with a synchronized clock. For example, one may be interested in determining the location of a gunshot heard by several microphones at known locations. There are many other scenarios in which this problem arises. For example, the bio-acoustic community is interested in using microphones to track the path of animals in nature in a non-intrusive way [\citenumber{rhinehart2020acoustic}]. In particular, recent work used this technique to identify patterns in the flight of some male birds [\citenumber{dutilleux2023chasing}]. The same principles are also used in a newly proposed muon-based navigation system for indoor and underground environments [\citenumber{tanaka2023first}].

\par{\bf Summary of the Results.}
In this paper\footnote{Preliminary work on this problem was included in an unpublished manuscript on a different topic [\citenumber{boutin2022multilateration}].}, we derive a simple solution method for the global positioning problem in $n$ dimensions with $m$ satellites (see \cref{pGPS}). Our solution is algebraic (direct) and assumes that the data received is exact. Depending on the arrangement of the satellites and the values of the times of arrivals of the signal, it either consists in solving a linear system of equations or in solving a quadratic equation for a 1D quantity (the offset) and then replacing it into a linear expression in order to obtain the user location. Assuming that the satellite locations are not coplanar, our solution method yields either one or two solutions (see \cref{tLateration}).
We describe the various possible cases in geometric terms in \cref{tUnique}.

In particular, this theorem gives a geometric criterion, in terms of the satellite positions and the user position, for when the global positioning problem has a unique solution. In a nutshell, it says that the solution is {\em not} unique if and only if the satellite positions lie on the same sheet of a hyperboloid of revolution, with the user position at one of its foci. 
Our results show that situations where the solution is not unique are ubiquitous: for any number of satellites, there exist situations where the satellites are not coplanar and yet the solution to the global positioning problem is not unique (see \cref{cUnique}). 
\cref{sExamples} gives several examples, including a situation with two solutions in dimension $n=3$ with $m=5$ satellites such that no four of them are coplanar. This example has no plane of symmetry, and thus disproves a recent conjecture by Hou [\citenumber{Hou:2022}]. We provide counterexamples to conjectures made by other authors in that section, too.
Numerical tests using randomly generated satellite positions in dimension $n=2, 3, 4$ with $m=n+1$ satellites
show that the proportion of user positions for which the solution is unique varies considerably for different satellite positions (see \cref{fDistribution}).

Practitioners have always worked under the assumption that the more satellites you have, the better. A common belief is that using at least $n+2$ satellites almost guarantees a unique solution, but this has never been proved. In \cref{sAlmost} of this paper, we formally confirm that this common assumption is correct by proving that non-coplanar configurations of $n+2$ satellites have a unique solution for almost all user positions (see \cref{tM5}).
Moreover, almost all configurations of at least $2 n + 2$ satellites are such that there is a unique solution for any user position (see \cref{tM8}).  Proving these last two results requires some tools from algebraic geometry.  \cref{tM9,tM10} describe explicit conditions on the satellite positions for which uniqueness for the user position is guaranteed. As far as we know, these last three results are the first ones stating uniqueness for all user positions. Our findings relating the number of satellites to the uniqueness of solutions are summarized in \cref{tConclusion}. 

\par{\bf Relation to Prior Work.} There is a significant amount of prior literature on the topic of global positioning, including textbooks such as \mycite{borre2012algorithms}. To our knowledge, the first algebraic solution of the global positioning problem was given by Schmidt in 1972 [\citenumber{Schmidt:1972}] for dimension $n=2$. In that work, the user position $\ve x$ is shown to lie on a straight line determined by the position of three satellites $\ve a_1,\ve a_2,\ve a_3$ and the differences between the times of arrival. The paper also presents a geometric criterion for the uniqueness question, which essentially constitutes the two-dimensional case of the criterion given in \cref{tUnique} of the present paper. See \cref{rSCA} for a more detailed comparison. However, \citename{Schmidt:1972}'s proof of the criterion makes use of an incorrect statement. See \cref{exSheet} for details and a counterexample. To the best of our knowledge, the present paper is the first to provide a correct proof of a geometric criterion for the uniqueness question.

Another algebraic solution is that of \mycite{Bancroft:1985}, which is well-cited in the more recent literature and applies to dimension $n=3$. That solution involves solving a quadratic equation. 
The simplicity of this approach makes it a popular choice for many applications. See \cref{rBancroft} for a detailed comparison of our solution procedure and the one in~[\citenumber{Bancroft:1985}]. \mycite{Krause:1987} gives another algebraic solution for the case of $m=4$ satellites in dimension $n=3$. 
An algebraic solution with lower computational cost for $m\geq 4$ satellites in dimension $n=3$ is given in \mycite{lundberg2001alternative}. Computational commutative algebra techniques (Gr\"obner bases and multipolynomial resultants) have been explored to obtain different algebraic solution of the GPS problem for a fixed dimension and number of satellites by \mycite{awange2002algebraic}.

On the other hand, there are several numerical (i.e., iterative) methods designed to handle noisy data. For example,
Li et al.~[\citenumber{li2010design}] propose a numerical approach to solve for the 2D position of a point on the globe (no altitude) from noisy measurements. Beck and Pan [\citenumber{Beck:Pan:2012}] discuss two different numerical approaches that apply to any number of satellites $m\geq 4$ in dimension $n=3$. Wang et al.~[\citenumber{wang2020gps}] improved one of these approaches in order to obtain better accuracy. 
Our method can provide an initial guess for any of these numerical approaches in the case where nothing is known a priori about the user position (e.g.~in space). More importantly, our uniqueness results provide a basis for a better understanding and prediction of the numerical behavior of these iterative algorithms, as ill-conditioning is expected for configurations near the ones for which the solution is not unique [\citenumber{spencer2007two},\citenumber{compagnoni2014comprehensive}]. More specifically, ill-conditioning is expected near the threshold between areas of nonuniqueness and uniqueness, as the erroneous solution can be interpreted as ``going off to infinity.''

There also exists some prior work on the question of uniqueness, to which a substantial part of this paper is devoted. We already discussed the part of \citename{Schmidt:1972}'s paper~[\citenumber{Schmidt:1972}] that concerns this question. The 1991 paper by \mycite{Abel:Chaffee:1991} debunks some beliefs about assured uniqueness that were held at the time, and, among other things, make some conjectures. One of those is disproved in the present paper (see \cref{ex3D}). In a later paper~[\citenumber{Chaffee:Abel:1994}], the same authors derive a criterion, in terms of the satellite positions and the user position, for unique solvability. This criterion only works for $m = 4$ satellites in three dimensions, and does not provide any geometric understanding. See \cref{rSCA} for a finer comparison between our criterion and the one from~[\citenumber{Chaffee:Abel:1994}]. A recent paper by \mycite{Hou:2022} gives a uniqueness criterion for the case $m = n + 2$, and makes an interesting conjecture about that case. However, the present paper provides a counterexample to that conjecture (see \cref{exFive,ex3D}).

Let us also mention the work of
\mycite{compagnoni2014comprehensive}, which gives a comprehensive analysis of the planar case with $m=3$ satellites. 
Our results complement and extend their results in a more general setting.

\par{\bf Acknowledgments.} A substantial part of this work was done
during a research stay of the authors at the Mathematical Sciences
Research Institute (MSRI) in Berkeley within the 2022 Summer Research
in Mathematics program. We thank the MSRI team for creating a uniquely
stimulating atmosphere at the institute. The program provided us a
with a perfect research environment and a chance to concentrate fully
on getting things done.  This work has also benefited from a research
stay of the two authors at the Banff International Research Station
for Mathematical Innovation and Discovery (BIRS) under the ``Research
in Teams'' program. We would like to thank BIRS for its hospitality
and for providing an optimal working environment. 
We are also very thankful to the 
Centro Internazionale per la Ricerca Matematica
(CIRM) in Trento for their friendly hospitality and support during a ``Research in Pairs" stay, which enabled us to bring this work to completion.
The software GeoGebra (www.geogebra.org) proved very useful for this research and gave us valuable geometric insight and intuition.
We thank Stefan
Wetkge and Timm Oertel for stimulating discussions.

\section{The solution method}
\label{sSolution}

Consider the following situation: A source at an unknown
position~$\ve x \in \RR^n$ emits a signal at an unknown time~$t$.
There are~$m$ sensors at known positions
$\ve a_1 \upto \ve a_m \in \RR^n$. They receive the signal at
times~$t_1 \upto t_m$. We choose the unit of time such that the signal
propagation speed becomes~$1$. So we have the fundamental equations
\begin{equation} \label{eqBasic}%
  \lVert\ve a_i - \ve x\rVert = t_i - t \qquad (i = 1 \upto m)
\end{equation}
The task is to work out the position~$\ve x$ and the emission
time~$t$. The very same equations arise if there are~$m$ sources
(e.g. satellites) at known positions~$\ve a_i$ emitting signals at
known times, which are then received by a device at an unknown
position~$\ve x$ (referred to as the user position). In
this case the~$t_i$ are the differences between the reception times
{\em according to the clock on the device} and the emission times
according to the (near-perfect) clocks on the sources, and~$t$ is the
unknown bias between the clock on the receiver and the clocks on the
sources. This is the \df{global positioning problem}. Since among the problems that involve solving~\eqref{eqBasic} for the unknowns~$\ve x$ and~$t$, the global positioning problem is the most familiar one, we will from now on think of the~$\ve a_i$ as the \df{satellite positions}, and of~$\ve x$ as the \df{user position}.

In the following derivation we will use the equations
\begin{equation} \label{eqAbsolute}%
  \lVert\ve a_i - \ve x\rVert = |t_i - t|, \qquad (i = 1 \upto m), \\
\end{equation}
which become equivalent to~\eqref{eqBasic} when combined with the
inequalities
\begin{equation} \label{eqIneq}%
  t_i \ge t, \qquad (i = 1 \upto m).
\end{equation}
Writing $L := \diag(-1,1 \upto 1) \in \RR^{(n+1) \times (n+1)}$,
$\tilde{\ve a_i} := \left(\begin{smallmatrix} t_i \\ \ve
    a_i\end{smallmatrix}\right)$ and
$\tilde{\ve x} := \left(\begin{smallmatrix} t \\ \ve
    x\end{smallmatrix}\right)$, we have
\begin{multline*}
  \lVert\ve a_i - \ve x\rVert^2 - (t_i - t)^2 = (\tilde{\ve a_i} -
  \tilde{\ve x})^T \cdot L \cdot (\tilde{\ve a_i} - \tilde{\ve x}) =
  \tilde{\ve a_i}^T L \tilde{\ve a_i} - 2 \tilde{\ve a_i}^T L \tilde{\ve
    x} + \tilde{\ve x}^T L \tilde{\ve x} = \\
  \lVert\ve a_i\rVert^2 - t_i^2 + 2 t_i t - 2 \ve a_i^T \ve x +
  \lVert\ve x\rVert^2 - t^2,
\end{multline*}
so~\eqref{eqAbsolute} is equivalent to
\begin{equation} \label{eqEquiv}%
  - 2 t_i t + 2 \ve a_i^T \ve x - \lVert\ve x\rVert^2 + t^2 =
  \lVert\ve a_i\rVert^2 - t_i^2 \qquad (i = 1 \upto m).
\end{equation}
We form the matrix
\begin{equation} \label{eqA}%
  A :=
  \begin{pmatrix}
    - 2 t_1 & 2 \ve a_1^T & -1 \\
    \vdots & \vdots & \vdots \\
    -2 t_m & 2 \ve a_m^T & -1
  \end{pmatrix} \in \RR^{m \times (n+2)},
\end{equation}
which contains only known quantities. With this, \eqref{eqEquiv} can
be expressed as a system of linear equations for the unknown
quantities:
\begin{equation} \label{eqMatrix}%
  A \cdot
  \begin{pmatrix}
    t \\
    \ve x \\
    \lVert\ve x\rVert^2 - t^2
  \end{pmatrix} =
  \begin{pmatrix}
    \lVert\ve a_1\rVert^2 - t_1^2 \\
    \vdots \\
    \lVert\ve a_m\rVert^2 - t_m^2
  \end{pmatrix}.
\end{equation}

Now we make the assumption that $A$ has rank~$n+2$,
so~\eqref{eqMatrix} has a unique solution. (The existence of at least
one solution follows from the fact that the point~$\ve x$ and number~$t$ exist.) In the case $m = n+2$ we can then simply
invert~$A$. More generally, for $m \ge n+2$, we use
the Moore-Penrose inverse
$A^+ = (A^T A)^{-1} A^T \in \RR^{(n+2) \times m}$, which is a left
inverse of $A$. So multiplying both sides of~\eqref{eqMatrix} with
$A^+$ and then with the matrix that deletes the last entry of the
vector $\left(\begin{smallmatrix} 1 \\ \ve x \\ \lVert\ve x\rVert^2 -
    t^2\end{smallmatrix}\right)$ yields
\begin{equation} \label{eqTxA}%
  \begin{pmatrix}
    t \\
    \ve x
  \end{pmatrix} =
  \begin{pmatrix}
    1 & \cdots & 0 & 0 \\
    \vdots & \ddots & \vdots & \vdots \\
    0 & \cdots & 1 & 0
  \end{pmatrix} \cdot A^+ \cdot \begin{pmatrix}
    \lVert\ve a_1\rVert^2 - t_1^2 \\
    \vdots \\
    \lVert\ve a_m\rVert^2 - t_m^2
  \end{pmatrix}.
\end{equation}
So we have obtained a unique solution of~\eqref{eqAbsolute}. This works if
$A$ has rank~$n+2$. In \cref{sAlmost} we will show that this is ``very
likely'' if $m \ge n+2$.

Now we consider the case that $A$ has rank~$\le n+1$, which is certain
to occur if $m \le n+1$, with the aim of solving~\eqref{eqAbsolute} in that case as well. We assume that the~$\ve a_i$ do not lie on a common affine
hyperplane. (So for
$n = 3$ this means that they are not coplanar, for $n = 2$ that they
are not collinear, and for $n = 1$ that they are not all equal.) This
assumption makes sense, since if the~$\ve a_i$ all lay in the same
hyperplane, then~\eqref{eqBasic} would almost never have a unique
solution (see \cref{rLateration}\ref{rLaterationA} below).
This assumption amounts to saying that the matrix
\begin{equation} \label{eqB}%
  B :=
  \begin{pmatrix}
    2 \ve a_1^T & -1 \\
    \vdots & \vdots \\
    2 \ve a_m^T & -1
  \end{pmatrix} \in \RR^{m \times (n+1)}
\end{equation}
has rank~$n+1$%
, so in particular we need $m \ge n+1$. Since $B$ is a submatrix of
$A$, we obtain $\rank(A) = n+1$. Now~\eqref{eqEquiv} can be restated
as
\begin{equation} \label{eqBxt}%
  B \cdot 
  \begin{pmatrix}
    \ve x \\
    \lVert\ve x\rVert^2 - t^2
  \end{pmatrix} =
  2 t \begin{pmatrix}
    t_1 \\
    \vdots \\
    t_m
  \end{pmatrix} + \begin{pmatrix}
    \lVert\ve a_1\rVert^2 - t_1^2 \\
    \vdots \\
    \lVert\ve a_m\rVert^2 - t_m^2
  \end{pmatrix},
\end{equation}
and multiplying by the Moore-Penrose inverse
$B^+ = (B^T B)^{-1} B^T \in \RR^{(n+1) \times m}$ yields
\begin{equation} \label{eqXUV}%
  \begin{pmatrix}
    \ve x \\
    \lVert\ve x\rVert^2 - t^2
  \end{pmatrix} = t \cdot
  \begin{pmatrix}
    \ve u \\
    2 \alpha
  \end{pmatrix} +
  \begin{pmatrix}
    \ve v \\
    \beta
  \end{pmatrix},
\end{equation}
where
\begin{equation} \label{eqUV}%
  \begin{pmatrix}
    \ve u \\
    2 \alpha
  \end{pmatrix} := 2 B^+
  \begin{pmatrix}
    t_1 \\
    \vdots \\
    t_m
  \end{pmatrix} \ \text{and} \
  \begin{pmatrix}
    \ve v \\
    \beta
  \end{pmatrix} := B^+
  \begin{pmatrix}
    \lVert\ve a_1\rVert^2 - t_1^2 \\
    \vdots \\
    \lVert\ve a_m\rVert^2 - t_m^2
  \end{pmatrix}.
\end{equation}
Extracting components from~\eqref{eqXUV}, we obtain the equations
\begin{equation} \label{eqRank4}%
  \ve x = t \ve u + \ve v \quad \text{and} \quad \bigl(\lVert\ve
  u\rVert^2-1\bigr) t^2 + 2 \bigl(\ve u^T \ve v-\alpha\bigr) t +
  \lVert\ve v\rVert^2 - \beta = 0.
\end{equation}
Observe that~$\ve u$, $\ve v$, $\alpha$, and~$\beta$ are all derived
from known quantities, so \eqref{eqRank4} can be resolved. This is how we can solve our problem in the case $\rank(A) =
n+1$. Since we derived~\eqref{eqRank4} from the
equations~\eqref{eqAbsolute}, some of its solutions may not satisfy
the original equations~\eqref{eqBasic}. But we can easily test the
inequalities~\eqref{eqIneq} for any solution of~\eqref{eqRank4}, so
there is hope that we can whittle down multiple solutions in this way
to just one. So we arrive at the following solution procedure, whose correctness will be formally proved in \cref{tLateration}.

\begin{procedure}[Solution of the global positioning problem] \label{pGPS}%
  \mbox{}
  \begin{description}
    \item[Input] Points~$\ve a_1 \upto \ve a_m \in \RR^n$, assumed not to lie on a common affine hyperplane, and~$t_1 \upto t_m \in \RR$. It is also assumed that~\eqref{eqBasic} is solvable.
    \item[Output] The set of solutions $(\ve x,t)$ of~\eqref{eqBasic}.
  \end{description}
  \begin{enumerate}[label=(\arabic*)]
    \item Compute the matrix $A$ according to~\eqref{eqA} and determine its rank.
    \item If $\rank(A) = n+2$ then compute $A^+ := (A^T A)^{-1} A^T$ and determine $\left(\begin{smallmatrix} t \\ \ve x\end{smallmatrix}\right) \in \RR^{n+1}$ according to~\eqref{eqTxA}. Output~$(\ve x,t)$ as the unique solution of~\eqref{eqBasic}. The solution procedure finishes here.
    \item If, on the other hand, $\rank(A) \le n+1$ then compute the matrix $B$ according to~\eqref{eqB}, and work out the vectors~$\ve u,\ve v$ and the scalars~$\alpha$ and~$\beta$ by using~\eqref{eqUV}.
    \item For every~$t \in \RR$ satisfying the equation
    \[
    \bigl(\lVert\ve u\rVert^2-1\bigr) t^2 + 2 \bigl(\ve u^T \ve v-\alpha\bigr) t + \lVert\ve v\rVert^2 - \beta = 0,
    \]
    perform step~\ref{pGPS4}.
    \begin{enumerate}[label=(\arabic*)]
    \setcounter{enumii}{\value{enumi}}
      \item \label{pGPS4} If~$t \le t_i$ for all $i = 1 \upto m$ then compute $\ve v := t \ve u + \ve v$ and output~$(\ve x,t)$ as a solution of~\eqref{eqBasic}.
    \end{enumerate}
    \setcounter{enumi}{\value{enumii}}
    \item The solution procedure finishes here, having output~$1$ or~$2$ solutions.
  \end{enumerate}
\end{procedure}

Notice that our solution procedure was directly derived from the equations~\eqref{eqBasic} that we want to solve. This means that every solution of~\eqref{eqBasic} will be found by \cref{pGPS}. However, the converse is another matter. Is it also true that every pair~$(\ve x,t)$ spit out by the procedure is really a solution of~\eqref{eqBasic}? On a finer scale, is it true that in the case $\rank(A) = n+1$, unique
solvability of the quadratic equation in the procedure is equivalent to unique solvability
of~\eqref{eqAbsolute}? These questions will be answered in the affirmative by \cref{tLateration}.
We find it helpful, if not essential, to state
this theorem and further results in a purely mathematical way, removing the
real-world context. Specifically, we will from now on work with the
following setup.

\begin{setup} \label{Setup}%
  Let $\ve a_1 \upto \ve a_m \in \RR^n$, with $n \ge 2$, be pairwise
  distinct points that do not lie on a common affine hyperplane, so
  $m \ge n+1$. Moreover, let~$\ve x \in \RR^n$ and let $t \in
  \RR$. For $i = 1 \upto m$ set
  \begin{equation} \label{eqSetup}
    t_i := \lVert\ve a_i - \ve x\rVert + t,
  \end{equation}
  where ``$\lVert\cdot\rVert$'' denotes the Euclidean norm. When we speak of solutions
  of~\eqref{eqBasic} or of~\eqref{eqAbsolute}, we will mean vectors
  $\left(\begin{smallmatrix} t' \\ \ve x'\end{smallmatrix}\right) \in
  \RR^{n+1}$ such that
  \[
    \lVert\ve a_i - \ve x'\rVert = t_i - t' \quad (i = 1 \upto m)
    \quad \text{or} \quad \lVert\ve a_i - \ve x'\rVert = |t_i -
    t'|\quad (i = 1 \upto m),
  \]
  respectively. Consider the matrices $A$ and $B$ given by the
  formulas~\eqref{eqA} and~\eqref{eqB}, and form
  $\ve u,\ve v \in \RR^n$ and~$\alpha,\beta \in \RR$ by the
  formulas~\eqref{eqUV}. Then a solution of~\eqref{eqRank4} will mean
  a vector
  $\left(\begin{smallmatrix} t' \\ \ve x'\end{smallmatrix}\right) \in
  \RR^{n+1}$ such that
  \[
    \bigl(\lVert\ve u\rVert^2-1\bigr) (t')^2 + 2 \bigl(\ve u^T \ve
    v-\alpha\bigr) t' + \lVert\ve v\rVert^2 - \beta = 0 \quad
    \text{and} \quad \ve x' = t' \ve u + \ve v.
  \]
\end{setup}

We now present the result announced above, which may be paraphrased by saying that \cref{pGPS} is optimal in the sense that it computes all solutions, but no ``spurious'' ones. It also contains the important fact that the quadratic equation in the procedure (and in~\eqref{eqRank4}) never degenerates completely.

\begin{theorem}[Correctness of \cref{pGPS}] \label{tLateration}%
  Assume the situation of \cref{Setup}, so
  $n+1 = \rank(B) \le \rank(A) \le n+2$. Then we have:
  \begin{enumerate}[label=(\alph*)]
  \item \label{tLaterationA} If $\rank(A) = n+2$,
    then~\eqref{eqAbsolute} has a unique solution given
    by~\eqref{eqTxA}, which is
    $\left(\begin{smallmatrix} t \\ \ve
        x\end{smallmatrix}\right)$. So~\eqref{eqBasic} has the same
    unique solution.
  \item \label{tLaterationB} If $\rank(A) = n+1$, then the
    equations~\eqref{eqAbsolute} and~\eqref{eqRank4} have the same
    solutions.
  \item \label{tLaterationC} Moreover, if $\rank(A) = n+1$, the
    coefficients of~$t^2$ and~$t$ in the quadratic equation
    in~\eqref{eqRank4} are not both zero. So~\eqref{eqRank4}, and by
    \cref{tLaterationB} also~\eqref{eqAbsolute}, has one or two
    solutions.
  \end{enumerate}
  Thus the output of \cref{pGPS} is exactly the solution set of~\eqref{eqBasic}, and there are one or two solutions.
\end{theorem}

The following lemma will be used in the proof of \cref{tLaterationC}
of the theorem and in \cref{sGeometry}.

\begin{lemma} \label{lSide}%
  In the situation of \cref{Setup}, assume $t = 0$. Then
  \begin{equation} \label{eqVBeta}%
    \ve v = \ve x \quad \text{and} \quad \beta = \lVert\ve x\rVert^2.
  \end{equation}
  Moreover, $\rank(A) < n+2$ if and only if
  \begin{equation} \label{eqUAAlpha}
    \langle\ve u,\ve a_i\rangle - \alpha = \lVert\ve a_i - \ve x\rVert \quad
    (i = 1 \upto m).
  \end{equation}
  No other pair of a vector and a scalar in the place of~$\ve u$
  and~$\alpha$ can satisfy~\eqref{eqUAAlpha}.
\end{lemma}

\begin{proof}
  The hypothesis, together with~\eqref{eqSetup} implies
  \begin{equation} \label{eqX0t0}%
    t_i = \lVert\ve a_i - \ve x\rVert \quad (i = 1 \upto m).
  \end{equation}
  We obtain
  \[
    \begin{pmatrix}
      \ve v \\ \beta
    \end{pmatrix} \!\underset{\eqref{eqUV},\eqref{eqX0t0}}{=}\!
    B^+
    \begin{pmatrix}
      2 \langle\ve x,\ve a_1\rangle - \lVert\ve x\rVert^2 \\
      \vdots \\
      2 \langle\ve x,\ve a_m\rangle - \lVert\ve x\rVert^2
    \end{pmatrix} \!\underset{\eqref{eqB}}{=}\! B^+
    B
    \begin{pmatrix}
      \ve x \\ \lVert\ve x\rVert^2
    \end{pmatrix} =
    \begin{pmatrix}
      \ve x \\ \lVert\ve x\rVert^2
    \end{pmatrix}.
  \]
  To prove the second assertion, assume $\rank(A) < n+2$. Then a look
  at~\eqref{eqA} and~\eqref{eqB} reveals that the last $n+1$
  columns of $A$ are linearly independent, but together with the first
  column they become linearly dependent. So there is a vector
  $\ve u' \in \RR^n$ and a scalar $\alpha' \in \RR$ such that
  \[
    A \cdot
    \begin{pmatrix}
      1 \\ \ve u' \\ 2 \alpha'
    \end{pmatrix}
    = 0.
  \]
  This can be rewritten as
  \begin{equation} \label{eqUsAs}%
    B \cdot
    \begin{pmatrix}
      \ve u' \\ 2 \alpha'
    \end{pmatrix}
    = 2
    \begin{pmatrix}
      t_1 \\ \vdots \\ t_m
    \end{pmatrix}
    \quad \text{or} \quad \langle\ve u',\ve a_i\rangle - \alpha'
    \!\!\underset{\eqref{eqX0t0}}{=}\!\!  \lVert\ve a_i - \ve x\rVert
    \quad (i = 1 \upto m).
  \end{equation}
  Conversely, having~$\ve u'$ and~$\alpha'$ satisfying~\eqref{eqUsAs}
  implies that $\rank(A) < n+2$. The second assertion, including the
  uniqueness statement, follows if we can deduce $\ve u' = \ve u$ and
  $\alpha' = \alpha$ from~\eqref{eqUsAs}. This can be done by
  multiplying the above matrix equation with $B^+$ and
  using~\eqref{eqUV}.
\end{proof}

\begin{proof}[Proof of \cref{tLateration}]
  For the proof of \cref{tLaterationA}, we assume $\rank(A) =
  n+2$. By~\eqref{eqSetup}, the vector
  $\left(\begin{smallmatrix} t \\ \ve x\end{smallmatrix}\right)$
  satisfies~\eqref{eqBasic} and thus
  also~\eqref{eqAbsolute}. Since~\eqref{eqTxA} was derived
  from~\eqref{eqAbsolute}, it also satisfies~\eqref{eqTxA}. By the
  same argument, any other vector
  $\left(\begin{smallmatrix} t' \\ \ve x'\end{smallmatrix}\right) \in
  \RR^{n+1}$ satisfying $\lVert\ve a_i - \ve x'\rVert = |t_i - t'|$
  for all~$i$ also satisfies~\eqref{eqTxA}, so
  \[
    \begin{pmatrix}
      t' \\ \ve x'
    \end{pmatrix} =
    \begin{pmatrix}
      1 & \cdots & 0 & 0 \\
      \vdots & \ddots & \vdots & \vdots \\
      0 & \cdots & 1 & 0
    \end{pmatrix} \cdot A^+ \cdot \begin{pmatrix}
      \lVert\ve a_1\rVert^2 - t_1^2 \\
      \vdots \\
      \lVert\ve a_m\rVert^2 - t_m^2
    \end{pmatrix} = \begin{pmatrix}
      t \\ \ve x
    \end{pmatrix}.
  \]
  This proves \cref{tLaterationA}, so from now on we may assume
  $\rank(A) = n+1$. By renumbering we can assume that the first
  $n + 1$ rows of $B$ are linearly independent, so the
  submatrix $B_0 \in \RR^{(n+1) \times (n+1)}$ formed by these
  rows is invertible. We also form the submatrix
  $A_0 \in \RR^{(n+1) \times (n+2)}$ of $A$ by selecting the first
  $n+1$ rows. Then
  $n+1 = \rank(B_0) \le \rank(A_0) \le \rank(A) = n+1$, so the
  remaining rows of $A$ are linear combinations of the first $n+1$
  rows. We express this by writing $A = C \cdot A_0$ with
  $C \in \RR^{m \times (n+1)}$ a matrix whose upper
  $(n+1) \times (n+1)$-part is the identity matrix. Extracting the
  first row and the remaining part from this gives
  \begin{equation} \label{eqACA}%
    \begin{pmatrix}
      t_1 \\ \vdots \\ t_m
    \end{pmatrix} = C \cdot
    \begin{pmatrix}
      t_1 \\ \vdots \\ t_{n+1}
    \end{pmatrix} \quad \text{and} \quad B = C \cdot B_0.
  \end{equation}
  By~\eqref{eqSetup}, the vector
  $\left(\begin{smallmatrix} t \\ \ve x\end{smallmatrix}\right)$
  satisfies~\eqref{eqAbsolute}. Since~\eqref{eqMatrix} was derived
  from~\eqref{eqAbsolute}, it also holds, and we get
  \begin{equation} \label{eqAtC}%
    \begin{pmatrix} \lVert\ve a_1\rVert^2 - t_1^2 \\ \vdots \\
      \lVert\ve a_m\rVert^2 - t_m^2
    \end{pmatrix} \!\underset{\eqref{eqMatrix}}{=}\! A
    \begin{pmatrix} t \\ \ve x \\ \lVert\ve x\rVert^2 - t^2
    \end{pmatrix} = C A_0
    \begin{pmatrix} t \\ \ve x \\ \lVert\ve x\rVert^2 - t^2
    \end{pmatrix} \!\underset{\eqref{eqMatrix}}{=}\! C \cdot
    \begin{pmatrix} \lVert\ve a_1\rVert^2 - t_1^2 \\ \vdots \\
      \lVert\ve a_{n+1}\rVert^2 - t_{n+1}^2
    \end{pmatrix}.
  \end{equation}
  Since $B^+$ is a left inverse of $B$, we also have
  \begin{equation} \label{eqAAC}%
    B B^+ C \!\!\underset{\eqref{eqACA}}{=}\!\! C B_0 B^+ C = C B_0
    B^+ C B_0 B_0^{-1} \!\!\underset{\eqref{eqACA}}{=}\!\! C B_0 B^+ B
    B_0^{-1} = C B_0 B_0^{-1} = C.
  \end{equation}
  For the proof of \cref{tLaterationB}, let us assume that
  $\left(\begin{smallmatrix} t' \\ \ve x'\end{smallmatrix}\right)$
  satisfies~\eqref{eqRank4}. This is equivalent to
  $\left(\begin{smallmatrix} \ve x' \\ \lVert\ve x'\rVert^2 -
      (t')^2 \end{smallmatrix}\right) = t' \cdot
  \left(\begin{smallmatrix} \ve u \\ 2 \alpha \end{smallmatrix}\right)
  + \left(\begin{smallmatrix} \ve v \\
      \beta \end{smallmatrix}\right)$, so
  \begin{multline*}
    B
    \begin{pmatrix}
      \ve x' \\ \lVert\ve x'\rVert^2 - (t')^2
    \end{pmatrix} \!\underset{\eqref{eqUV}}{=} B B^+
    \Bigl(2 t'
    \begin{pmatrix}
      t_1 \\ \vdots \\ t_m
    \end{pmatrix} +
    \begin{pmatrix}
      \lVert\ve a_1\rVert^2 - t_1^2 \\
      \vdots \\
      \lVert\ve a_m\rVert^2 - t_m^2
    \end{pmatrix}\Bigr)
    \!\!\underset{\eqref{eqACA},\eqref{eqAtC}}{=} \\
    B B^+ C \Bigl(2 t'
    \begin{pmatrix}
      t_1 \\ \vdots \\ t_{n+1}
    \end{pmatrix} +
    \begin{pmatrix}
      \lVert\ve a_1\rVert^2 - t_1^2 \\
      \vdots \\
      \lVert\ve a_{n+1}\rVert^2 - t_{n+1}^2
    \end{pmatrix}\Bigr)
    \!\!\underset{\eqref{eqACA}-\eqref{eqAAC}}{=}\!\! 2 t'
    \begin{pmatrix}
      t_1 \\ \vdots \\ t_m
    \end{pmatrix} +
    \begin{pmatrix}
      \lVert\ve a_1\rVert^2 - t_1^2 \\
      \vdots \\
      \lVert\ve a_m\rVert^2 - t_m^2
    \end{pmatrix}.
  \end{multline*}
  So~\eqref{eqBxt} is satisfied. But that equation just
  restates~\eqref{eqEquiv}, which in turn is equivalent
  to~\eqref{eqAbsolute}. So we have seen that if
  $\left(\begin{smallmatrix} t' \\ \ve x'\end{smallmatrix}\right)$
  satisfies~\eqref{eqRank4}, then it also
  satisfies~\eqref{eqAbsolute}. The converse also holds
  since~\eqref{eqRank4} was derived from~\eqref{eqAbsolute}. This
  concludes the proof of \cref{tLaterationB}.

  \Cref{tLaterationC} can be restated by saying that~\eqref{eqRank4}
  has neither infinitely many solutions nor no solution at
  all. By~\ref{tLaterationB} it is equivalent to say this about the
  solutions of~\eqref{eqAbsolute}. We know that~\eqref{eqAbsolute} has
  at least one solution, so no solutions is impossible. Now the number
  of solutions, i.e., the number of vectors
  $\left(\begin{smallmatrix} t' \\ \ve x'\end{smallmatrix}\right)$
  satisfying $\lVert\ve a_i - \ve x'\rVert = |t_i - t'|$ for
  $i = 1 \upto m$, with~$t_i$ given by~\eqref{eqSetup}, does not
  change if a real number is subtracted from~$t$ and consequently from
  all~$t_i$. In particular, we can subtract~$t$ from the~$t_i$ and~$t$
  without changing the number of solutions. For the proof of
  \cref{tLaterationC} we can therefore assume $t = 0$, so we can use
  \cref{lSide}.

  By~\eqref{eqVBeta}, the quadratic equation in~\eqref{eqRank4}
  becomes
  $\bigl(\lVert\ve u\rVert^2-1\bigr) t^2 + 2 \bigl(\langle\ve u,\ve
  x\rangle - \alpha\bigr) t = 0$, so assuming that there are
  infinitely many solutions means $\langle\ve u,\ve x\rangle = \alpha$
  and $\lVert\ve u\rVert = 1$. With~\eqref{eqUAAlpha} this implies
  \[
    \langle\ve u,\ve a_i - \ve x\rangle = \lVert\ve a_i - \ve x\lVert
    = \lVert\ve u\rVert \cdot \lVert\ve a_i - \ve x\lVert,
  \]
  so the Cauchy–Schwarz inequality is actually an equality, implying
  that the pair of vectors~$\ve a_i - \ve x$ and~$\ve u$ is linearly
  dependent. So $\ve a_i = \ve x + \alpha_i \ve u$ with
  $\alpha_i \in \RR$. This shows that the $\ve a_i$ are
  collinear. Since $n \ge 2$, this contradicts the hypothesis that
  they do not lie on a common affine hyperplane, and the proof is
  complete.
\end{proof}

Let us briefly look at the cases that were excluded in \cref{Setup}:
that the dimension~$n$ is one, and that the vectors~$\ve a_i$ lie on a
common affine hyperplane.

\begin{rem} \label{rLateration}%
  \begin{enumerate}[label=(\alph*)]
  \item \label{rLaterationA} What happens in the one-dimensional case?
    It is well known (see \mycite{Hou:2022}) that if~$\ve x$ lies on
    the ``same side'' of all $\ve a_i$, i.e., if $\ve x \le \ve a_i$
    for all~$i$ or $\ve x \ge \ve a_i$ for all~$i$, then
    determining~$\ve x$ is impossible, in the sense
    that~\eqref{eqBasic} has infinitely many solutions. On the other
    hand, if~$\ve x$ lies anywhere between the~$\ve a_i$, then its
    determination is possible, i.e., \eqref{eqBasic} has a unique
    solution. So it makes sense that the final step in the proof of
    \cref{tLateration}\ref{tLaterationC} used $n \ge 2$. In fact, in
    the case $n = 1$ and~$\ve x$ lying on the same side of
    the~$\ve a_i$, our matrix $A$ has rank~$2 = n+1$, and the
    coefficients of the quadratic equation in~\eqref{eqRank4} turn out
    to be all~$0$.
    
    

  \item \label{rLaterationB} Let us also consider the case that the
    $\ve a_i$ do lie on a common affine hyperplane $H \subset
    \RR^n$. Let~$\ve x'$ be the point obtained by reflecting~$\ve x$
    at $H$. Then
    $\left(\begin{smallmatrix} t \\ \ve x'\end{smallmatrix}\right)$ is
    a solution of~\eqref{eqBasic}, and $\ve x' \ne \ve x$ unless
    $\ve x \in H$. So in this situation~\eqref{eqBasic} almost never
    has a unique solution. \remend
  \end{enumerate}
  \renewcommand{\remend}{}
\end{rem}

\begin{rem} \label{rBancroft}%
  The most cited,
  algebraic solution of the global positioning problem was given by
  \mycite{Bancroft:1985} in~\citeyear{Bancroft:1985}. Let us point out
  some differences between his approach and ours.
  \begin{itemize}
  \item \citename{Bancroft:1985} reaches a quadratic equation even if
    there are more than~$n + 1$ sensors, and does not give a
    ``linear'' formula such as~\eqref{eqTxA} in this case. He states
    that of the two solutions of his quadratic equation, only one will
    satisfy~\eqref{eqBasic}. As already remarked by
    \mycite{Abel:Chaffee:1991}, this is not necessarily true.
A similar issue arises with the computational commutative algebra techniques used in \mycite{awange2002algebraic}, which  yield algebraic formulations 
involving solving one polynomial of degree two.

  \item \citename{Bancroft:1985} assumes that the matrix
    \begin{equation} \label{eqBancroft}%
      \begin{pmatrix}
        \ve a_1^T & t_1 \\
        \vdots & \vdots \\
        \ve a_m^T & t_m
      \end{pmatrix} \in \RR^{m \times (n + 1)}
    \end{equation}
    (called ``$A$'' in~[\citenumber{Bancroft:1985}]) has
    rank~$n+1$. (In fact, he works in dimension $n = 3$.) But, as
    already remarked by \mycite{Chaffee:Abel:1994}, even a coordinate
    change that makes~$\ve a_1$ become~$\ve 0$ and~$t_1$ become~$0$
    results in this matrix failing to have full rank. (Also see
    \cref{exCone} in this paper.) On the other hand, our condition
    that the matrix $B$ from~\eqref{eqB} have rank $n + 1$ seems more
    natural and is coordinate-independent. 
  \remend
  \end{itemize}
  \renewcommand{\remend}{}
\end{rem}

\section{Interlude: quadrics with a focus} \label{sQuadric}%

Our next goal is to provide a geometric characterization of
when our fundamental equations~\eqref{eqBasic} have a unique solution. For this purpose, we first need to deal with some facts related to quadrics. Although this is just elementary
geometry, we were unable to find a suitable reference, except for dimension three (see \mycite{zwillinger2018crc}). We start by recalling the focus-directrix property.

A point set $Q \subseteq \RR^n$, with $n \ge 2$, is said to satisfy
the \df{focus-directrix property} if it is given as
\begin{equation} \label{eqQ0}%
  Q = \bigl\{ P \in \RR^n \ \bigr| \ \lVert P - F\rVert = e \cdot
  d(H,P)\bigr\},
\end{equation}
where $F \in \RR^n$ is a point (the \df{focus}),
$H \subset \RR^n$ is an affine hyperplane (the \df{directrix}),
$e \in \RR_{\ge 0}$ is a number (the \df{eccentricity}), and
$d(H,P)$ denotes the distance between $H$ and $P$. We make the
additional hypothesis that $Q$ contains non-collinear
points. Extending the concept, we also regard the $(n-1)$-sphere with
any positive radius around $F$ as satisfying the focus-directrix
property with $e = 0$, thinking of a directrix at infinity. (For
shortness, we will from now on just say {\em sphere} for the
$(n-1)$-sphere.)  

If $e > 0$, it is convenient to choose a normal vector~$\ve u$ on $H$
with length $\lVert\ve u\rVert = e$. Then $H$ is given as
$H = \bigl\{ P \in \RR^n \mid \langle\ve u,P\rangle = \alpha\bigr\}$
for some $\alpha \in \RR$. With this, we obtain
\begin{equation} \label{eqQ1}%
  Q = \bigl\{ P \in \RR^n \ \bigr| \ \lVert P - F\rVert = |\langle\ve
  u,P\rangle - \alpha|\bigr\}.
\end{equation}
It is helpful that for $\ve u = \ve 0$ this defines the sphere with
radius~$|\alpha|$ around~$F$, so by including the possibility
$\ve u = \ve 0$ we need not consider the sphere as a separate case any
longer. $Q$~splits into two subsets $Q_+$ and $Q_-$ with
\begin{equation} \label{eqPlusMinus}%
  Q_\pm := \bigl\{ P \in \RR^n \ \bigr| \ \lVert P - F\rVert =
  \pm\bigl(\langle\ve u,P\rangle - \alpha\bigr)\bigr\},
\end{equation}
the intersections of $Q$ with the two closed half-spaces defined by
the hyperplane $H$. However, $Q_+$ and $Q_-$ can only be both nonempty
if $Q$ is disconnected or if $F \in Q$. Each of $Q_+$ and $Q_-$, if
nonempty, is called a \df{sheet} of $Q$. If $F \in Q$ it is contained
in both sheets. There is nothing intrinsically positive or negative
about the sheets, since they can be swapped by substituting~$\ve u$
by~$-\ve u$ and~$\alpha$ by~$-\alpha$. Comparing the definition of
$Q_+$ with~\eqref{eqUAAlpha} already shows how intimately the
focus-directrix property is connected to the topic of this paper.

It is also convenient to choose a unit vector
$\tilde{\ve u} \in \RR^n$ with $\ve u = e \tilde{\ve u}$, so
$\tilde{\ve u} = e^{-1} \ve u$ if $e > 0$, and $\tilde{\ve u}$
arbitrary if~$e = 0$. Apart from the eccentricity, $Q$ has another
parameter, which is obtained by intersecting $Q$ with the affine plane
through $F$ with normal vector~$\tilde{\ve u}$. Using~\eqref{eqQ1}, we
find the intersection to be
\[
  \bigl\{F + \ve v \ \bigr| \ \ve v \in \RR^n, \langle\tilde{\ve
    u},\ve v\rangle = 0, \lVert\ve v\rVert = |\langle\ve u,F\rangle -
  \alpha|\bigr\},
\]
so it is an $(n-2)$-sphere around $F$ with radius $|l|$, where we have
set
\begin{equation} \label{eqSemilatus}%
  l := \langle\ve u,F\rangle - \alpha.
\end{equation}
The radius~$|l|$ is called the \df{semilatus rectum}, and is the other
numerical parameter of $Q$, complementing the
eccentricity. \cref{fQuadric} contains an illustration of two
instances of the focus-directrix property that occur in dimension two.

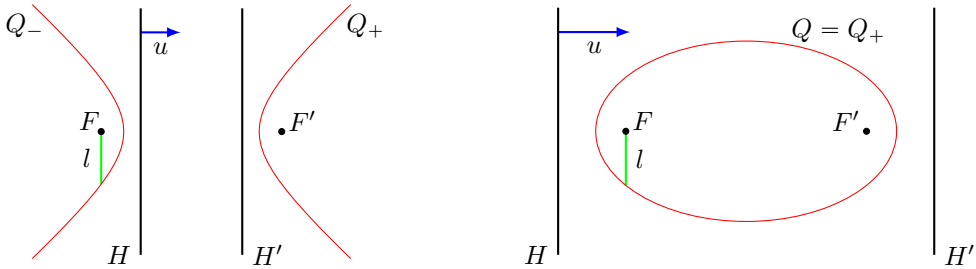
\begin{figure}[htbp]
  \centering
  \begin{tikzpicture}[scale=0.4]
    \draw[thick,green] (-3,0)--(-3,-1.75)
    node[midway,left,black]{$l$};%
    \tikzset{>=latex}; 
    \draw[->,thick,blue] (-1.6875,3.3)--(4/3-1.6875,3.3)
    node[midway,below,black]{$u$};%
    \pgfmathsetmacro{\a}{sqrt(318.9375/81)}%
    \pgfmathsetmacro{\b}{sqrt(318.9375/63)}%
    \draw[red] plot[domain=-1.5:1.5] ({{\b*cosh(\x)},\a*sinh(\x)})
    node[below,xshift=2mm,black]{$Q_+$};%
    \draw[red] plot[domain=-1.5:1.5] ({-\b*cosh(\x)},{\a*sinh(\x)})
    node[below,xshift=-1mm,black]{$Q_-$};%
    \filldraw (3,0) circle (1mm)
    node[right,yshift=4,xshift=-1]{$F'$};%
    \filldraw (-3,0) circle (1mm) node[left,yshift=4,xshift=3]{$F$};%
    \draw[thick] (1.6875,-4.1)--(1.6875,4.1) node[at
    start,right]{$H'$};%
    \draw[thick] (-1.6875,-4.1)--(-1.6875,4.1) node[at
    start,left]{$H$};%
  \end{tikzpicture}%
  \hspace{15mm} %
  \begin{tikzpicture}[scale=1.2]
    \draw[thick,green] (-4/3,0)--(-4/3,-3/5)
    node[midway,right,black]{$l$};%
    \tikzset{>=latex}; 
    \draw[->,thick,blue] (-25/12,1.1)--(4/5-25/12,1.1)
    node[midway,below,black]{$u$};%
    \draw[red] (0,0) ellipse (5/3 and 1)
    node[yshift=38,xshift=35,black]{$Q = Q_+$};%
    \filldraw (4/3,0) circle (0.33mm)
    node[left,yshift=4,xshift=1]{$F'$};%
    \filldraw (-4/3,0) circle (0.33mm)
    node[right,yshift=4,xshift=-1]{$F$};%
    \draw[thick] (25/12,-1.37)--(25/12,1.37) node[at
    start,right]{$H'$};%
    \draw[thick] (-25/12,-1.37)--(-25/12,1.37) node[at
    start,left]{$H$};%
  \end{tikzpicture}
  
  \caption{Hyperbolae and ellipses have two foci $F$ and $F'$ with
    corresponding directices $H$ and $H'$. Here we have $e = 4/3$ and
    $|l| = 7/4$ for the hyperbola, and $e = 4/5$ and $|l| = 3/5$ for
    the ellipse.}
  \label{fQuadric}
\end{figure}

What do point sets $Q$ with the focus-directrix property look like?
Setting $P' := P - F$, the defining equation~\eqref{eqQ1} for $Q$
becomes $\lVert P'\rVert = \bigl| \langle\ve u,P'\rangle + l\bigr|$,
and squaring both sides yields the equivalent equation
\begin{equation} \label{eqQ2}%
  \lVert P'\rVert^2 - \langle\ve u,P'\rangle^2 - 2 l \langle\ve
  u,P'\rangle - l^2 = 0.
\end{equation}
This shows that $Q$ is a quadric. From now on we will speak of a \df{quadric with a focus},
or specifically a \df{quadric with focus $F$}, to mean a point set
satisfying the focus-directrix property. Of course this need not mean
that $F$ is the {\em only} focus (see \cref{fQuadric}). (In fact, the quadrics with focus constitute a subclass of the quadrics of revolution). 

To find out which types of quadrics have a focus, write
$P'_\perp := P' - \langle\tilde{\ve u},P'\rangle \cdot \tilde{\ve u}$
for the part of $P'$ that is perpendicular to $\tilde{\ve u}$. With
this, \eqref{eqQ2} becomes
\begin{equation} \label{eqQ3}%
  (1 - e^2) \langle\tilde{\ve u},P'\rangle^2 + \lVert P'_\perp\rVert^2
  -2 e l \langle\tilde{\ve u},P'\rangle - l^2 = 0.
\end{equation}
First assume $e \ne 1$ and set
$P'' := P' - \frac{e l}{1 - e^2} \tilde{\ve u}$. Then
$\langle\tilde{\ve u},P''\rangle = \langle\tilde{\ve u},P'\rangle -
\frac{e l}{1 - e^2}$, so
\[
  (1 - e^2)\langle\tilde{\ve u},P''\rangle^2 = (1 - e^2)
  \langle\tilde{\ve u},P'\rangle^2 - 2 e l \langle\tilde{\ve
    u},P'\rangle + \frac{e^2 l^2}{1 - e^2} \! = \! (1 - e^2)
  \langle\tilde{\ve u},P'\rangle^2 - 2 e l \langle\tilde{\ve
    u},P'\rangle - l^2 + \frac{l^2}{1- e^2}.
\]
Since $P''_\perp = P'_\perp$, multiplying~\eqref{eqQ3} by $1 - e^2$,
yields the equivalent equation
\begin{equation} \label{eqQ4}%
  \left((1 - e^2)\langle\tilde{\ve u},P''\rangle\right)^2 + \sgn(1 -
  e^2) \left(\sqrt{|1 - e^2|} \cdot \lVert P''_\perp\rVert\right)^2 -
  l^2 = 0.
\end{equation}
Extending~$\tilde{\ve u}$ by vectors $\ve v_1 \upto \ve v_{n-1}$ to an
orthonormal basis, we obtain cartesian coordinates
$z = \langle\tilde{\ve u},P''\rangle$ and
$x_i = \langle\ve v_i,P''\rangle$. With these~\eqref{eqQ4} becomes
\begin{equation} \label{eqQ5}%
  \left((1 - e^2) z\right)^2 + \sgn(1 - e^2) \sum_{i=1}^{n-1}
  \left(\sqrt{|1 - e^2|} \cdot x_i\right)^2 - l^2 = 0.
\end{equation}
Before we go on, let us consider the remaining case $e = 1$. Setting
$P'' := \sgn(l) \cdot \bigl(P' + \frac{l}{2} \tilde{\ve u}\bigr)$
gives
$2 |l| \cdot \langle\tilde{\ve u},P''\rangle = 2 l\langle\tilde{\ve
  u},P' + \frac{l}{2} \tilde{\ve u}\rangle = 2 e l\langle\tilde{\ve
  u},P'\rangle + l^2$, so~\eqref{eqQ3} becomes
\begin{equation} \label{eqQ6}%
  \lVert P''_\perp\rVert^2 - 2 |l| \cdot \langle\tilde{\ve
    u},P''\rangle = 0, \quad \text{or} \quad - 2 |l| z +
  \sum_{i=1}^{n-1} x_i^2 = 0,
\end{equation}
using coordinates as above.

Now we can distinguish four cases.
\begin{enumerate}[label= \bf Case \arabic*:, ref= \arabic*, align=left]
\item \label{case1} $e < 1$. This implies $l \ne 0$ since
  otherwise~\eqref{eqQ5} has $z = x_1 = \cdots = x_{n-1} = 0$ as the
  only solution, contradicting the hypothesis about non-collinear
  points in $Q$. With
  $a := \frac{|l|}{1 - e^2} \ge b := \frac{|l|}{\sqrt{1 - e^2}}$,
  from~\eqref{eqQ5} we obtain the standard equation, shown in the
  first row of \cref{taQuadrics},
  of a \df{prolate spheroid} ($0 < e < 1$, $a > b$) or a sphere ($e = 0$,
  $a = b$). It has one sheet and one (in case of a sphere) or two (in
  case of a spheroid) foci.
\item \label{case2} $e > 1$ and $l \ne 0$. Setting
  $a := \frac{|l|}{e^2 - 1}$ and $b := \frac{|l|}{\sqrt{e^2 - 1}}$, we
  get the second equation shown in \cref{taQuadrics},
  which defines a \df{hyperboloid of revolution of two sheets}. The sheets
  are given by $z > 0$ and $z < 0$. Again, there are two foci.
\item \label{case3} $e > 1$ and $l = 0$. Dividing~\eqref{eqQ5} by
  $(e^2 - 1)^2$ and setting $a := \sqrt{e^2 - 1}$ yields the third
  equation shown in \cref{taQuadrics}, which
  gives a \df{cone of revolution}. It has one focus (equal to the vertex)
  and two sheets, which are given by $z \ge 0$ and $z \le 0$.
\item \label{case4} $e = 1$. This implies $l \ne 0$ since
  otherwise~\eqref{eqQ6} would imply that all points are
  collinear. With $a := \sqrt{|l|}$, we obtain the last equation in
  \cref{taQuadrics}.
  This defines a \df{paraboloid of revolution}, which has one focus and one
  sheet.
\end{enumerate}

With this, we have classified all quadrics with a focus. The results
are summarized in \cref{taQuadrics}.

\begin{table}[ht]
\[
  \renewcommand{\arraystretch}{1.6}
  \begin{array}{|c|c|c|c|c|}
    \hline
    \text{\bf case} & \text{\bf conditions} & \text{\bf equation} &
    \text{\bf name} & \text{\bf number of sheets}
    \\
    \hline
    \ref{case1} & e < 1, |l| > 0 & \frac{z^2}{a^2} + \sum_{i=1}^{n-1}
    \frac{x_i^2}{b^2} = 1 & \parbox[c][9mm]{40mm}{\centering spheroid
    ($e > 0$, $a > b$) \\ or sphere ($e = 0$, $a = b$)} & 1 \\
    \hline
    \ref{case2} & e > 1, |l| > 0 & \frac{z^2}{a^2} - \sum_{i=1}^{n-1}
    \frac{x_i^2}{b^2} = 1 & \parbox[c][9mm]{40mm}{\centering
    hyperboloid of revolution of two sheets} & 2 \\
    \hline
    \ref{case3} & e > 1, |l| = 0 & z^2 = \sum_{i=1}^{n-1} \frac{x_i^2}{a^2} &
    \text{cone of revolution} & 2 \\
    \hline
    \ref{case4} & e = 1, |l| > 0 & 2 z = \sum_{i=1}^{n-1} \frac{x_i^2}{a^2} &
    \text{paraboloid of revolution} & 1
    \\
    \hline
  \end{array}
\]
\caption{Classification of quadrics with a focus in any dimension} \label{taQuadrics}
\end{table}

Of course in 2D a spheroid is called an ellipse, a hyperboloid is known as a hyperbola, a cone
becomes a pair of intersecting lines, and a paraboloid goes by the
name parabola.

\section{The geometry of nonuniqueness} \label{sGeometry}%

Given the positions~$\ve a_i$ and the times~$t_i$, one can run \cref{pGPS}, which will reveal whether
there is a unique solution or not. But to gain some understanding of
the uniqueness question---and to prove further results with this
understanding---we wish to give a geometric characterization of
when~\eqref{eqBasic} has a unique solution. This is contained in
\cref{tUnique}. In fact, the theorem
does more than that. It gives a geometric interpretation for each case
that may occur in the solution procedure from \cref{sSolution}. Recall
that \cref{pGPS} preferentially uses~\eqref{eqTxA}, which is
available if the matrix $A$, defined in~\eqref{eqA}, has
rank~$n+2$. Only if $\rank(A) < n + 2$ does the procedure resort to
solving~\eqref{eqRank4}, which involves a quadratic equation, and
which is equivalent to~\eqref{eqAbsolute}. In the expected case of two
solutions, the inequalities~\eqref{eqIneq} are then used to try to
discard one of the solutions. If this is not possible, this means
that~\eqref{eqBasic} has two solutions. We state the result in purely
mathematical terms.

\begin{theorem}[Geometric characterization] \label{tUnique}%
  Assume the situation of \cref{Setup}. Then~\eqref{eqBasic} has two
  solutions if and only if the points~$\ve a_1 \upto \ve a_m$ lie on
  the same sheet of a hyperboloid of revolution with focus~$\ve
  x$. Otherwise~\eqref{eqBasic} has exactly one solution, which is
  $\left(\begin{smallmatrix} t \\ \ve x\end{smallmatrix}\right)$. More
  precisely, we have:
  \begin{enumerate}[label=(\alph*)]
  \item \label{tUniqueA} The matrix $A \in \RR^{m \times (n+2)}$,
    given by~\eqref{eqA}, has rank $< n + 2$ if and only if
    the~$\ve a_i$ lie on the same sheet of a quadric $Q$ with
    focus~$\ve x$. In this case, there exists only one such $Q$, and
    we have the following facts~\ref{tUniqueB}--\ref{tUniqueF}:
  \item \label{tUniqueB} The eccentricity of $Q$ is given by
    $e = \lVert\ve u\rVert$, with $\ve u \in \RR^n$ defined
    by~\eqref{eqUV}. If $e \ne 1$, then the discriminant of the
    quadratic equation in~\eqref{eqRank4} is $4 l^2$, with~$|l|$ the
    semilatus rectum of $Q$.
  \item \label{tUniqueC} $Q$ is a spheroid or a sphere if and only
    if~\eqref{eqAbsolute}, and therefore also~\eqref{eqRank4}, has two
    solutions, but~\eqref{eqBasic} has only one. So in this case the
    inequalities~\eqref{eqIneq} allow discarding one solution
    $\left(\begin{smallmatrix} t' \\ \ve x'\end{smallmatrix}\right)$
    of~\eqref{eqAbsolute}. Moreover, $\ve x'$ is a focus of $Q$, and
    $\ve x' = \ve x$ if and only if $Q$ is a sphere. An example is
    shown in \cref{fSheet}.
  \item \label{tUniqueD} $Q$ is a hyperboloid of revolution if and
    only if~\eqref{eqBasic}, and therefore also~\eqref{eqAbsolute}
    and~\eqref{eqRank4}, have two solutions. For the solution
    $\left(\begin{smallmatrix} t' \\
        \ve x'\end{smallmatrix}\right) \ne \left(\begin{smallmatrix} t
        \\ \ve x\end{smallmatrix}\right)$ of~\eqref{eqBasic} we have
    that $\ve x'$ is the other focus of $Q$, and $|t - t'|$ is equal
    to the distance between the two sheets of $Q$. An example is shown
    in \cref{fFive}.
  \item \label{tUniqueE} $Q$ is a cone of revolution if and only if
    the quadratic equation in~\eqref{eqRank4} has discriminant
    zero. In this case~\eqref{eqRank4}, and therefore
    also~\eqref{eqAbsolute} and~\eqref{eqBasic}, have exactly one
    solution. Examples are shown in \cref{fCone}.
  \item \label{tUniqueF} $Q$ is a paraboloid of revolution if and
    only if the coefficient of~$t^2$ in the quadratic equation
    in~\eqref{eqRank4} is zero. So in this case the equation is in
    fact linear, thus~\eqref{eqRank4}, and therefore
    also~\eqref{eqAbsolute} and~\eqref{eqBasic}, have exactly one
    solution.
  \end{enumerate}
\end{theorem}

\begin{proof}
  Before dealing with \crefrange{tUniqueA}{tUniqueF}, we show that we
  may assume that $t = 0$. In fact, subtracting any real number
  from~$t$ (as introduced in \cref{Setup}) and consequently also from
  the~$t_i$ leaves the following quantities unchanged:
  \begin{itemize}
  \item the~$\ve a_i$ and~$\ve x$, and therefore anything that can be
    said about their geometry;
  \item the vector $\ve x'$ in a solution
    $\left(\begin{smallmatrix} t' \\ \ve x'\end{smallmatrix}\right)$
    of~\eqref{eqBasic} or~\eqref{eqAbsolute}, and so by
    \cref{tLateration} also of~\eqref{eqRank4};
  \item the number of solutions of~\eqref{eqBasic}
    or~\eqref{eqAbsolute}, and so by \cref{tLateration} also
    of~\eqref{eqRank4};
  \item the truth or falsehood of~\eqref{eqIneq};
  \item the rank of $A$, since subtracting the same number from
    all~$t_i$ amounts to subtracting a multiple of the last column of
    $A$ from the first column of $A$;
  \item the vector~$\ve u$, since $B^+$, being a left inverse of $B$,
    must have its upper~$n$ rows with coefficient sum equal to zero,
    so~$\ve u$, defined by~\eqref{eqUV}, does not change when the same
    number is subtracted from all~$t_i$;
  \item the discriminant of the quadratic equation in~\eqref{eqRank4},
    since the discriminant is formed from the leading coefficient and
    the difference between its zeros, which both remain unchanged.
  \end{itemize}
  Since this list encompasses all objects about which the theorem
  makes assertions, we may, without loss of generality, subtract~$t$
  itself from~$t$ and from the~$t_i$, which amounts to assuming
  $t = 0$, so the equation~\eqref{eqSetup} defining the~$t_i$ becomes
  \begin{equation} \label{eqTi}%
    t_i = \lVert\ve a_i - \ve x\rVert.
  \end{equation}
  After this preparation, let us turn our attention to
  \cref{tUniqueA}. If $\rank(A) < n + 2$, then \cref{lSide} gives us
  the equation~\eqref{eqUAAlpha}, which by~\eqref{eqPlusMinus} means
  that the~$\ve a_i$ lie on one the sheet $Q_+$ of the quadric with
  focus~$\ve x$, given by the vector~$\ve u$ and the
  scalar~$\alpha$. Conversely, assume that the $\ve a_i$ lie on the
  same sheet $Q'_+$ or $Q'_-$ of some quadric $Q'$ with focus~$\ve x$,
  given by a vector~$\ve u'$ and a scalar~$\alpha'$. By
  replacing~$\ve u'$ and~$\alpha'$ by their negatives if necessary, we
  may assume $\ve a_i \in Q'_+$. This means that~\eqref{eqUAAlpha},
  with~$\ve u$ and~$\alpha$ replaced by~$\ve u'$ and~$\alpha'$,
  holds. By \cref{lSide}, this implies $\ve u' = \ve u$
  and~$\alpha' = \alpha$, so $Q' = Q$, and with this~\eqref{eqUAAlpha}
  implies $\rank(A) < r + 2$. This concludes the proof
  of~\ref{tUniqueA}, and we know that the quadric $Q$ is given
  by~\eqref{eqQ1}, with $F$ replaced by~$\ve x$, and~$\ve u$
  and~$\alpha$ given by~\eqref{eqUV}.

  It follows directly that the eccentricity of $Q$ is
  $e = \lVert\ve u\rVert$. Moreover, using~\eqref{eqVBeta}, we find
  that the quadratic equation in~\eqref{eqRank4} becomes
  $\bigl(\lVert\ve u\rVert^2 -1\bigr) t^2 + 2\bigl(\langle\ve u,\ve
  x\rangle - \alpha\bigr) t = 0$. If $e \ne 1$, its discriminant is
  $4 (\langle\ve u,\ve x\rangle - \alpha\bigr)^2$, which
  by~\eqref{eqSemilatus} equals $4 l^2$, where $|l|$ is the semilatus
  rectum of $Q$. So we have proved \cref{tUniqueB}. But now by looking
  at \cref{taQuadrics} we see that $Q$ is a cone of revolution if and
  only if the discriminant of the quadratic equation is zero, which
  proves \cref{tUniqueE}; and that $Q$ is a paraboloid of revolution
  if and only if the coefficient of~$t^2$ vanishes, which proves
  \cref{tUniqueF}.

  Now assume that $Q$ is a spheroid or a hyperboloid of revolution
  (the case of a sphere is easier and will be dealt with later). The
  standard equations in \cref{taQuadrics} show that by mapping the
  coordinate~$z$ to~$-z$ we obtain another focus~$\ve x' \ne \ve x$
  and a corresponding directrix $H'$ of $Q$, with $H'$ parallel to
  $H$. If $Q$ is a spheroid, then all points $P \in Q$ lie between $H$
  and $H'$, and by~\eqref{eqQ0} we obtain
  \begin{equation} \label{eqPlus}%
    \lVert P - \ve x\rVert + \lVert P - \ve x'\rVert = e \cdot
    \Bigl(d(H,P) + d(H',P)\Bigr) = e \cdot d(H,H') =: t'.
  \end{equation}
  On the other hand, if $Q$ is a hyperboloid of revolution, $H$ and
  $H'$ lie between $Q_+$ and $Q_-$, so for $P \in Q_+$ we obtain
  \begin{equation} \label{eqMinus}%
    \lVert P - \ve x\rVert - \lVert P - \ve x'\rVert = e \cdot
    \Bigl(d(H,P) - d(H',P)\Bigr) = e \cdot \bigl(\pm d(H,H')\bigr) =:
    t',
  \end{equation}
  where the sign (i.e., the meaning of ``$\pm$'') depends on whether
  $H$ or $H'$ lies closer to $Q_+$. In particular, the appropriate one
  of Equations~\eqref{eqPlus} or~\eqref{eqMinus} holds for
  $P = \ve a_i$, and we get
  \begin{equation} \label{eqAxs}%
    \lVert\ve a_i - \ve x'\rVert = \pm\bigl(\lVert\ve a_i - \ve
    x\rVert - t'\bigr) \underset{\eqref{eqTi}}{=} \pm (t_i - t') \quad
    (i = 1 \upto m),
  \end{equation}
  where this time ``$\pm$'' means~''$+$'' in the case of a
  hyperboloid, and ``$-$'' in the case of a spheroid. In both cases,
  \eqref{eqAxs} means that
  $\left(\begin{smallmatrix} t' \\ \ve x'\end{smallmatrix}\right)$ is
  a solution of~\eqref{eqAbsolute}, and since $\ve x \ne \ve x'$, it
  is different from
  $\left(\begin{smallmatrix} t \\ \ve x\end{smallmatrix}\right)$. But
  since $\lVert\ve a_i - \ve x'\rVert > 0$ (because no focus of $Q$ is
  a point of $Q$ in the cases considered here), \eqref{eqAxs} also
  shows that
  $\left(\begin{smallmatrix} t' \\ \ve x'\end{smallmatrix}\right)$
  satisfies~\eqref{eqBasic} if and only if $Q$ is a hyperboloid. (It
  is a remarkable dichotomy that if $t_i - t < 0$ for one~$i$, then
  this is automatically true for all~$i$.)

  Finally, we treat the case that $Q$ is a sphere, so all
  $\lVert\ve a_i - \ve x\rVert$ are equal. Then
  \[
    \lVert\ve a_i - \ve x\rVert \underset{\eqref{eqTi}}{=} t_i = -
    \bigl(t_i - 2 \lVert\ve a_i - \ve x\rVert\bigr),
  \]
  so $\left(\begin{smallmatrix} t' \\ \ve x'\end{smallmatrix}\right)$,
  with $t' = 2 \lVert\ve a_i - \ve x\rVert$ and $\ve x' = \ve x$, is a
  solution of~\eqref{eqAbsolute} but not of~\eqref{eqBasic}.

  To prove the statement on $|t - t'|$ in \cref{tUniqueD} we use
  \cref{tUniqueB}. This tells us that the quadratic equation with
  solutions~$t$ and~$t'$ has discriminant $4 l^2$ and leading
  coefficient $e^2 - 1$. So the solutions have distance
  $2 |l|/|e^2 - 1|$. But a look at Cases~\ref{case1} and~\ref{case2}
  on page~\pageref{case1} shows that this equals $2 a$, with~$a$ the
  semi-major axis of $Q$. So indeed $|t - t'|$ is the distance between
  the vertices of $Q$. Notice that this also holds in the case of
  \cref{tUniqueC}, where for a sphere one can choose any pair of
  antipodes as vertices.

  We have now considered every case of what $Q$ can be, and verified
  the predictions of the theorem for all of them. In particular, the
  missing \cref{tUniqueC,tUniqueD}, including the ``if and only if''
  statements, have also been proved.
\end{proof}

\begin{rem} \label{rSCA}%
  Criteria for the unique solvability of special cases
  of~\eqref{eqBasic} were already given as early
  as~\citeyear{Schmidt:1972} by \mycite{Schmidt:1972} and later by
  \mycite{Chaffee:Abel:1994}. In fact, \citename{Schmidt:1972}
  considers the same quadric $Q$ as our \cref{tUnique}, and
  essentially makes the statements in \cref{tUniqueC,tUniqueD} of the
  theorem. Moreover, in the notation
  of~[\citenumber{Chaffee:Abel:1994}], \citename{Chaffee:Abel:1994}'s
  criterion reads ``$\langle a^\perp,a^\perp\rangle \le 0$'', and
  converting it to our situation and notation gives
  ``$\lVert\ve u\rVert \le 1$'', so there is a close relationship
  to~[\citenumber{Schmidt:1972}] and to our \cref{tUnique}. The new
  contributions made by \cref{tUnique}, in comparison
  to~[\citenumber{Schmidt:1972}] and~[\citenumber{Chaffee:Abel:1994}],
  are:
  \begin{itemize}
  \item \citename{Schmidt:1972}'s result is restricted to the case
    of~$m = 3$ satellites in 
    dimension~$n = 2$, with an excursion to~$m = 4$ and $n =
    3$. \citename{Chaffee:Abel:1994}'s result is restricted to the
    case~$m = 4$ and~$n = 3$. In contrast, \cref{tUnique} covers
    arbitrarily many sensors in any dimension~$n$.
  \item \cref{tUnique} corrects an inaccuracy
    in~[\citenumber{Schmidt:1972}]: in the case that $Q$ is a
    hyperboloid, \citename{Schmidt:1972} does not state that the
    points~$\ve a_i$ need to lie on the same sheet of $Q$. As
    demonstrated in \cref{exSheet} below, this condition is essential even
    for defining $Q$.
  \item \cref{tUnique} corrects a small error in both
    \citename{Schmidt:1972}'s and \citename{Chaffee:Abel:1994}'s
    criteria that occurs if the quadric $Q$ from our theorem is a cone
    of revolution. In fact, this case is missing
    from~[\citenumber{Schmidt:1972}]. Moreover, in the case of a cone
    the system~\eqref{eqBasic} has a unique solution, but
    $\lVert\ve u\rVert > 1$ or, in the notation
    of~[\citenumber{Chaffee:Abel:1994}],
    $\langle a^\perp,a^\perp\rangle >
    0$. So~[\citenumber{Chaffee:Abel:1994}] incorrectly says that the
    solution is not unique. \cref{exCone} below provides more details.
  \item \cref{tUnique} is not only a generalization of the previous
    criteria to arbitrary~$m$ and~$n$, but also a refinement since all
    cases that may occur in the solution procedure are given a precise
    geometric interpretation.
  \item Apart from providing uniqueness criteria, \cref{tUnique} also contains
    quantitative statements. In particular, the quantification of
    $|t - t'|$ in \cref{tUniqueD} may be useful in global positioning applications
    for discarding a solution if something is known about the accuracy
    of the user clock. 
  \item Finally, the proof in~[\citenumber{Schmidt:1972}] relies on an incorrect statement (see \cref{exSheet}), so our result can be considered to provide a proof (and correction) of Schmidt's criterion. \remend
  \end{itemize}
  \renewcommand{\remend}{}
\end{rem}

It is now clear that there are examples of arbitrarily many points~$\ve a_i$ (standing for satellite positions) such that the global positioning problem still has no unique
solution: just place them on one sheet of a hyperboloid of revolution,
and place~$\ve x$ (the GPS device) at its focus. Recall that some
given points in $\RR^n$ are said to be in \df{linear general position}
if any $(k-1)$-dimensional affine subspace of $\RR^n$ contains at
most~$k$ of these points, for $k = 1 \upto n$. There has been some
speculation in the literature~[\citenumber{Abel:Chaffee:1991},
\citenumber{Chaffee:Abel:1994}, \citenumber{Hou:2022}] that for certain numbers of satellites, being
in general linear position may be enough to ensure uniqueness. The
following consequence of \cref{tUnique} puts such hopes to rest.

\begin{cor} \label{cUnique}%
  For every~$m \in \NN$ there exist
  points~$\ve x, \ve a_1 \upto \ve a_m \in \RR^n$ in general linear
  position such that~\eqref{eqBasic} does not have a unique solution.
\end{cor}

\begin{proof}
  We need to show that one can put arbitrarily many points in general
  linear position on one sheet of a hyperboloid of revolution. While
  this is hardly surprising, the formal proof nevertheless requires
  some effort.
  
  Take $Q \subset \RR^n$ to be the hyperboloid of revolution given as
  the set of zeros of the polynomial
  $f := x_1^2 - \sum_{i=2}^n x_i^2 - 1$. We claim that the vanishing
  ideal is
  \[
    \mathcal I(Q) := \bigl\{g \in \RR[x_1 \upto x_n] \mid g \
    \text{vanishes on} \ Q\bigr\} = f \cdot \RR[x_1 \upto x_n].
  \]
  The inclusion ``$\supseteq$'' is clear, so for the reverse inclusion
  let $g \in \mathcal I(Q)$. Regarding~$x_1$ as main variable, we can
  perform division with remainder by~$f$ and obtain
  $g = f \cdot q + g_1 x_1 + g_0$ with $q \in \RR[x_1 \upto x_n]$ and
  $g_0,g_1 \in \RR[x_2 \upto x_n]$. Then
  $g_1 x_1 + g_0 \in \mathcal I(Q)$. If we assume $g_1 \ne 0$, there
  exist $\xi_2 \upto \xi_n \in \RR$ such that
  $g_1(\xi_2 \upto \xi_n) \ne 0$, so only
  $\xi_1 = -g_0(\xi_2 \upto \xi_n)/g_1(\xi_2 \upto \xi_n)$ can yield a
  point on $Q$ with $x_2 = \xi_2 \upto x_n = \xi_n$, which is not
  true: there are two such points. So $g_1 = 0$. Now if we assume
  $g_0 \ne 0$, again there are $\xi_2 \upto \xi_n \in \RR$ with
  $g_0(\xi_2 \upto \xi_n) \ne 0$. This implies that there is no point
  on $Q$ with $x_2 = \xi_2 \upto x_n = \xi_n$, which is not true. We
  conclude $g_0 = g_1 = 0$, so
  $g = f q \in f \cdot \RR[x_1 \upto x_n]$.

  One sheet $Q_+$ of $Q$ is given by the condition $x_1 > 0$. We take
  $\ve x$ to be a focus of $Q$ and prove the existence
  of~$\ve a_1 \upto \ve a_m \in Q_+$ by induction on~$m$. There is
  nothing to show for~$m = 0$. So assume~$m > 0$ and that
  $\ve a_1 \upto \ve a_{m-1} \in Q_+$ have been picked such that
  $\ve x,\ve a_1 \upto \ve a_{m-1}$ are in general linear
  position. Every subset
  $S \subseteq \{\ve x,\ve a_1 \upto \ve a_{m-1}\}$ of size
  $|S| \le n$ is contained in some affine hyperplane, so there exits a
  polynomial $h_S \in \RR[x_1 \upto x_n]$ of total degree~$1$ that
  vanishes on $S$. Write $h_S^*$ for the polynomial obtained from
  ~$h_S$ by substituting~$x_1$ with~$-x_1$, and set
  \[
    h := \prod_{\substack{S \subseteq \{\ve x,\ve a_1 \upto \ve
        a_{m-1}\} \\ \text{with} \ |S| \le n}} (h_S h_S^*).
  \]
  Being irreducible and of total degree~$2$, $f$ does not divide~$h$,
  so $h \notin \mathcal I(Q)$. Thus we have $\ve a_m \in Q$ with
  $h(\ve a_m) \ne 0$. If the $x_1$-coordinate of $\ve a_m$ is
  negative, we can replace it by its opposite since~$h$ is invariant
  under this replacement. So we may assume $\ve a_m \in Q_+$.

  It remains to show that $\ve x,\ve a_1 \upto \ve a_m$ are in general
  linear position, so let $U \subset \RR^n$ be a $(k-1)$-dimensional
  affine subspace with $1 \le k \le n$. By induction,
  $S := U \cap \{\ve x,\ve a_1 \upto \ve a_{m-1}\}$ has at most~$k$
  elements. If $S$ has fewer than~$k$ elements, then
  $|U \cap \{\ve x,\ve a_1 \upto \ve a_m\}| \le k$, the desired
  bound. Since $h_S$ vanishes on $S$, $S$ is contained in the subspace
  $U' := \{P \in U \mid h_S(P) = 0\} \subseteq U$. So if $|S| = k$,
  then by induction $U'$ cannot have dimension $< k - 1$, hence
  $U' = U$. But since $h_S(\ve a_m) \ne 0$, this implies
  $\ve a_m \notin U$. Therefore also in this case we have
  $|U \cap \{\ve x,\ve a_1 \upto \ve a_m\}| \le k$.

  This finishes the induction proof that there exist points
  $\ve a_1 \upto \ve a_m \in Q_+$ such that
  $\ve x,\ve a_1 \upto \ve a_m$ are in linear general position. If
  $m \ge n+1$, this implies that the~$\ve a_i$ are not contained in an
  affine hyperplane, so we are in the situation of \cref{Setup}, and
  \cref{tUnique} tells us that~\eqref{eqBasic} does not have a unique
  solution. On the other hand, if $m \le n$, then the~$\ve a_i$ lie in
  an $(m-1)$-dimensional affine subspace $U \subset \RR^n$ which by
  the general position property cannot contain~$\ve x$. $U$ can be
  enlarged to a hyperplane still not containing~$\ve x$. Now
  \cref{rLateration}\ref{rLaterationB} tells us that in this case,
  too, there is no unique solution.
\end{proof}

\section{Examples} \label{sExamples}

This section is devoted to providing examples. Some of them are intended as illustrations of the theory developed in the previous section, and some constitute counterexamples to conjectures or statements in the literature. For ease of drawing, we concentrate
on the case of dimension $n = 2$.

\begin{ex} \label{exSheet}%
  The situations shown in \cref{fSheet} demonstrate that the phrase
  ``\ldots on the same sheet of \ldots'' appearing twice in
  \cref{tUnique} is essential.
  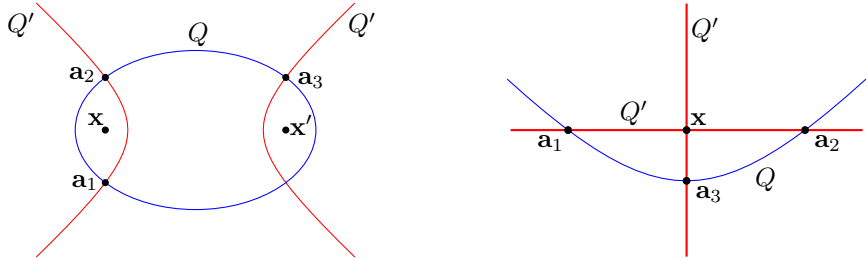
\begin{figure}[htbp]
    \centering
    \begin{tikzpicture}[scale=0.4]
      \pgfmathsetmacro{\a}{sqrt(318.9375/81)}%
      \pgfmathsetmacro{\b}{sqrt(318.9375/63)}%
      \draw[red] plot[domain=-1.5:1.5] ({{\b*cosh(\x)},\a*sinh(\x)})
      node[below,xshift=1mm,black]{$Q'$};%
      \draw[red] plot[domain=-1.5:1.5] ({-\b*cosh(\x)},{\a*sinh(\x)})
      node[below,xshift=-2mm,black]{$Q'$};%
      \draw[blue] (0,0) ellipse (4 and 2.64575)
      node[yshift=36,xshift=1,black]{$Q$};%
      \filldraw (3,0) circle (1mm)
      node[right,yshift=2,xshift=-2]{$\ve x'$};%
      \filldraw (-3,0) circle (1mm)
      node[left,yshift=4,xshift=3]{$\ve x$};%
      \filldraw (-3,-1.75) circle (1mm)
      node[left,yshift=0,xshift=1]{$\ve a_1$};%
      \filldraw (-3,1.75) circle (1mm) node[left,xshift=0,yshift=1]{$\ve
        a_2$};%
      \filldraw (3,1.75) circle (1mm)
      node[right,xshift=1,yshift=-1]{$\ve a_3$};%
    \end{tikzpicture}%
    \hspace{15mm} %
    \begin{tikzpicture}[scale=0.9]
      \draw[thick,red] (0,1.13)--(0,4.87)
      node[right,black,yshift=-10,xshift=-2]{$Q'$};%
      \draw[thick,red] (-2.6,3)--(2.6,3)
      node[midway,above,black,xshift=-20,yshift=-2]{$Q'$};%
      \pgfmathsetmacro{\a}{sqrt(318.9375/81)}%
      \pgfmathsetmacro{\b}{sqrt(318.9375/63)}%
      \draw[blue] plot[domain=-1.1:1.1] ({\a*sinh(\x)},{\b*cosh(\x)})
      node[below,yshift=-30,xshift=-38,black]{$Q$};%
      \filldraw (0,3) circle (0.5mm)
      node[right,yshift=4,xshift=-2]{$\ve x$};%
      \filldraw (-1.75,3) circle (0.5mm)
      node[left,yshift=-5,xshift=2]{$\ve a_1$};%
      \filldraw (1.75,3) circle (0.5mm)
      node[right,yshift=-5,xshift=0]{$\ve a_2$};%
      \filldraw (0,2.25) circle (0.5mm)
      node[right,xshift=0,yshift=-5]{$\ve a_3$};%
    \end{tikzpicture}%
    \caption{Both quadrics $Q$ and $Q'$ share the focus~$\ve x$ and
      the points~$\ve a_i$. But only $Q$ has the~$\ve a_i$ on the same
      sheet.}
    \label{fSheet}
  \end{figure}
  
  The first picture shows a hyperbola $Q'$ and an ellipse $Q$ sharing
  both foci and three points. (In fact, they share a fourth point, but
  this is irrelevant here.) Since the points are on
  different sheets of $Q'$, $Q$ is the quadric that \cref{tUnique}
  speaks about, and so~\eqref{eqBasic} has a unique solution. In the
  second picture, $Q$ and $Q'$ are a hyperbola and a cone (which in 2D
  is a pair of intersecting lines), both sharing the focus~$\ve x$ and
  three points. But~$\ve a_1$ and~$\ve a_2$ lie on different sheets of
  the cone, since no half-plane given by a diagonal line
  through~$\ve x$ contains both~$\ve a_1$ and~$\ve a_2$. So here the
  ``correct'' quadric is the hyperbola, which leads to the conclusion
  of nonunique solutions.

  Let us mention here that the first picture in \cref{fSheet}
  provides a counterexample to a statement on page~833
  of the paper of \mycite{Schmidt:1972}: ``To specify the locations of three points on a conic,
  and a straight line as its major axis, is to completely specify the 
  conic.'' In fact, counterexamples abound: a conic is indeed specified by three points and a focus, so just fixing the major axis allows the focus to move on it, giving a $1$-dimensional family of conics.
\end{ex}

\begin{ex} \label{exCone}%
  Here we look at the case where the satellites lie on the same sheet of a cone with vertex (=
  focus)~$\ve x$, so the quadratic equation in~\eqref{eqRank4} has
  discriminant zero. This is depicted in dimensions two and three in
  \cref{fCone}.
  
  \begin{figure}[htbp]
    \centering
    \begin{tikzpicture}[scale=0.8]
      \draw[blue,thick] (0,0)--(2,2) node[black,right,near
      end]{$Q_+$};%
      \draw[blue,thick] (0,0)--(-2,2) node[black,left,near
      end]{$Q_+$};%
      \draw[thick,red] (0,0)--(2,-2) node[black,right,near
      end]{$Q_-$};%
      \draw[thick,red] (0,0)--(-2,-2) node[black,left,near
      end]{$Q_-$};%
      \filldraw (0,0) circle (0.5mm) node[right]{$\ve x$};%
      \filldraw (0.5,0.5) circle (0.5mm)
      node[left,xshift=1,yshift=2]{$\ve a_1$};%
      \filldraw (1.2,1.2) circle (0.5mm)
      node[left,xshift=1,yshift=2]{$\ve a_2$};%
      \filldraw (-1,1) circle (0.5mm)
      node[right,xshift=1,yshift=2]{$\ve a_3$};%
    \end{tikzpicture}%
     \hspace{15mm} %
     \begin{tikzpicture}[xscale=2.2,yscale=-2.2]
       \begin{scope}
         \draw[clip] ({-sqrt(91)/10},9/100) coordinate (A)%
         arc[x radius=1, y radius=0.3, start angle=180-17.46, end
         angle=360+17.46] -- (0,1) -- cycle;%
         \fill[gray] (0,1) -- (-1.3628,-0.3) -- (1.3628,-0.3) --
         cycle;%
         \foreach \i in {0,1,...,100} {%
           \fill[white!\i!gray] (0,1) -- (-1.3628+\i*0.006814,-0.3) --
           +(0.008,0) -- cycle;%
           \fill[gray!\i!white] (0,1) -- (-0.6814+\i*0.006814,-0.3) --
           +(0.008,0) -- cycle;%
           \fill[gray!\i!black!50!gray] (0,1) --
           (0.6814-\i*0.006814,-0.3) -- +(-0.008,0) -- cycle;%
           \fill[black!\i!gray!50!black] (0,1) --
           (0.6814+\i*0.006814,-0.3) -- +(-0.008,0) -- cycle;%
         }%
          \draw (A) arc[x radius=1, y radius=0.3, start
          angle=180-17.46, end angle=17.46];%
         \draw[clip] ({-sqrt(91)/10},9/100) coordinate (A)%
         arc[x radius=1, y radius=0.3, start angle=180-17.46, end
         angle=360+17.46] -- (A) arc[x radius=1, y radius=0.3, start
         angle=180-17.46, end angle=17.46] -- cycle;%
         \fill[gray] (0,1) -- (-1.3628,-0.3) -- (1.3628,-0.3) --
         cycle;%
         \foreach \i in {0,...,100} {%
           \fill[white!\i!gray] (0,1) -- (0.6814+0.6814-\i*0.006814,-0.3) --
           +(-0.008,0) -- cycle;%
           \fill[gray!\i!white] (0,1) -- (0.6814-\i*0.006814,-0.3) --
           +(-0.008,0) -- cycle;%
           \fill[gray!\i!black!50!gray] (0,1) --
           (-0.6814+\i*0.006814,-0.3) -- +(0.008,0) -- cycle;%
           \fill[black!\i!gray!50!black] (0,1) --
           (0.6814-1.3628-\i*0.006814,-0.3) -- +(-0.008,0) -- cycle;%
         }%
       \end{scope}
       \draw (0.7,0.4) node[below]{$Q_+$};
       \draw (0,1) node[below,xshift=4,yshift=1]{$\ve x$};
       \filldraw (0.6,0.1) ellipse[rotate=-50,x radius=0.2mm,y
       radius=0.13mm] node[right,yshift=2]{$\ve a_1$};%
       \filldraw (-0.3,-0.1) ellipse[rotate=75,x radius=0.2mm,y
       radius=0.16mm] node[right,yshift=2]{$\ve a_2$};%
       \filldraw (0.0,0.2) ellipse[rotate=90,x radius=0.14mm,y
       radius=0.14mm] node[right,yshift=2]{$\ve a_3$};%
       \filldraw (-0.3,0.6) ellipse[rotate=45,x radius=0.2mm,y
       radius=0.12mm] node[right,yshift=2]{$\ve a_4$};%
     \end{tikzpicture}
     \caption{All~$\ve a_i$ lie on the same sheet of a cone with
       vertex~$\ve x$. The right-hand picture is not true to scale
       relative to the given numerical example.}
    \label{fCone}
  \end{figure}
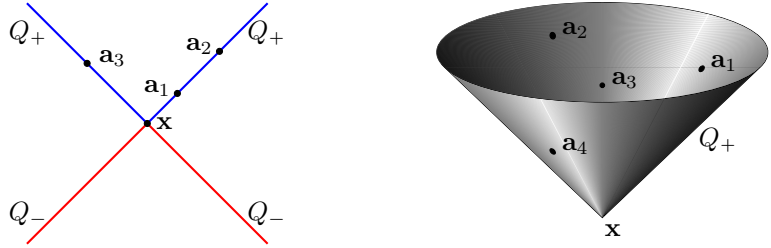

  To give a numerical example in dimension three, consider the points
  \[
    \ve x = (0,0,0), \ \ve a_1 = (3,4,5), \ \ve a_2 = (5,12,13), \ \ve
    a_3 = (8,15,17), \ \text{and} \ \ve a_4 = (7,24,25).
  \]
  Since the~$\ve a_i$ are chosen as Pythagorean triples, they lie on a
  cone with vertex~$\ve x$. So \cref{tUnique}\ref{tUniqueE} predicts a
  unique solution, due to discriminant zero in the quadratic
  equation. To test this, we set $t := 0$, so $t_1 = 5 \sqrt 2$,
  $t_2 = 13 \sqrt 2$, $t_3 = 17 \sqrt 2$, and $t_4 = 25 \sqrt 2$. Then
  our solution procedure easily yields the quantities
  $\ve u = (0,0,\sqrt 2)$, $\ve v = (0,0,0)$, and
  $\alpha = \beta = 0$. So the quadratic equation in~\eqref{eqRank4}
  becomes $t^2 = 0$, which has indeed discriminant zero.

  Let us also take the opportunity to apply \citename{Bancroft:1985}'s
  method~[\citenumber{Bancroft:1985}]. This is more cumbersome, since
  it requires that the matrix given in~\eqref{eqBancroft} has full
  rank, which in this case fails. So, somewhat arbitrarily, we need to
  change~$t$. With $t = \sqrt 2$, \citename{Bancroft:1985}'s method
  yields the quadratic equation $1/2 \lambda^2 + \lambda + 1/2 =
  0$. This also has discriminant zero, again confirming our
  prediction.

  We have already mentioned in \cref{rSCA} that
  \citename{Chaffee:Abel:1994}'s criterion goes wrong in the case
  where $Q$ is a cone, so let us check this, too, for this
  example. The matrix $A$ defined in Proposition~2
  from~[\citenumber{Chaffee:Abel:1994}] is easily determined to be
  \[
    A =
    \begin{pmatrix}
      2 & 8 & 8 & 8 \sqrt 2 \\
      5 & 11 & 12 & 12 \sqrt 2 \\
      4 & 20 & 20 & 20 \sqrt 2
    \end{pmatrix} \qquad \text{(notation
      from~[\citenumber{Chaffee:Abel:1994}])}.
  \]
  Its nullspace if generated by $a^\perp := (0,0,\sqrt 2,-1)$, and the
  Lorentz inner product is $\langle a^\perp,a^\perp\rangle = 1$. In
  fact, in~[\citenumber{Chaffee:Abel:1994}] $a^\perp$ is set to
  be a unit vector, but this does not change the positivity of
  $\langle a^\perp,a^\perp\rangle$. So \citename{Chaffee:Abel:1994}'s
  criterion says that there are two solutions, which is not true.
\end{ex}

\begin{ex} \label{exFive}%
  This is a two-dimensional example of \cref{cUnique}. We have placed
  five points~$\ve a_i$ (standing for satellite positions) on the same sheet (= branch) of a hyperbola with
  foci~$\ve x$ and~$\ve x'$. One of them, say~$\ve x$, stands for the
  user position. So by \cref{tUnique}\ref{tUniqueD}, both foci are
  solutions of the fundamental equations~\eqref{eqBasic}. \cref{fFive} shows
  the situation.
  
  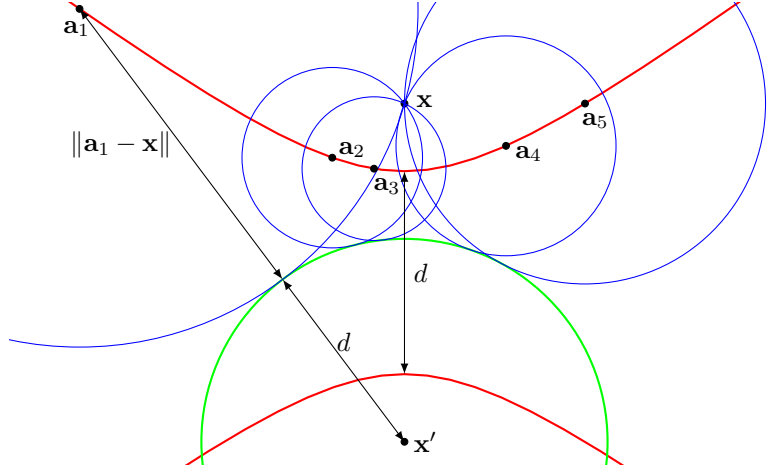
\begin{figure}[htbp]
    \centering
    \begin{tikzpicture}[scale=0.15]
      \pgfmathsetmacro{\r}{0.3};%
      \pgfmathsetmacro{\a}{12};%
      \pgfmathsetmacro{\b}{9};%
      \tikzset{>=latex}; 
      \begin{scope}
        \clip (-35,-17) rectangle (35,24);%
        \draw[thick,red] plot[domain=-2:2]
        ({\a*sinh(\x)},{\b*cosh(\x)});%
        \draw[thick,red] plot[domain=-2:2]
        ({\a*sinh(\x)},{-\b*cosh(\x)});%
        \filldraw (0,15) circle (\r)
        node[right,xshift=1,yshift=1]{$\ve x$};%
        \filldraw (0,-15) circle (\r) node[right] (X) {$\ve x'$};%
        \draw[thick,green] (0,-15) circle (18);%
        
        \filldraw (-28.8,23.4) circle (\r)
        node[below,yshift=-1,xshift=-1]{$\ve a_1$};%
        \draw[blue] (-28.8,23.4) circle (30);%
        \filldraw (-6.4,10.2) circle (\r)
        node[right,xshift=-1,yshift=3]{$\ve a_2$};%
        \draw[blue] (-6.4,10.2) circle (8);%
        \filldraw (-2.7,9.225) circle (\r)
        node[below,xshift=5]{$\ve a_3$};%
        \draw[blue] (-2.7,369/40) circle (51/8);%
        \filldraw (9,11.25) circle (\r)
        node[right,yshift=-3]{$\ve a_4$};%
        \draw[blue] (9,11.25) circle (9.75);%
        \filldraw (16,15) circle (\r) node[below,xshift=4]{$\ve
          a_5$};%
        \draw[blue] (16,15) circle (16);%
        \draw[thin,<->] (0,9)--(0,-9) node[midway,right]{$d$};
        \draw[thin,<->] (0,-15)--(-10.8,-0.6) node[midway,above]{$d$};
        \draw[thin,<->] (-28.8,23.4)--(-10.8,-0.6)
        node[midway,left]{$\lVert\ve a_1 - \ve x\rVert$};
     \end{scope}
    \end{tikzpicture}%
    \caption{The two red sheets have distance~$d$. The blue circles
      around the~$\ve a_i$ 
      through~$\ve x$ are tangent with the green circle. This shows
      $\lVert\ve a_i - \ve x\rVert + d = \lVert\ve a_i - \ve
      x'\rVert$.}
    \label{fFive}
  \end{figure}
  
  \cref{fFive} ``visually'' confirms our result, by showing
  $\lVert\ve a_i - \ve x'\rVert = \lVert\ve a_i - \ve x\rVert + d$ for
  all~$i$, which implies
  $\lVert\ve a_i - \ve x'\rVert = t_i - t + d = t_i - t'$ for
  $t' := t - d$, where~$t$ is the correct time.

  Let us provide the coordinates of the points in this example. The
  following table gives the coordinates of~$\ve x$, $\ve x'$ and
  the~$\ve a_i$, as well as the distances of each point to~$\ve x$ and
  to~$\ve x'$.
  \[
    \begin{array}{c|cc|ccccc}
      \text{Point} & \ve x & \ve x' & \ve a_1 & \ve a_2 & \ve a_3 &
                                                                    \ve
                                                                    a_4
      & \ve a_5 \\ \hline
      $x$-\text{coordinate} & 0 & 0 & -28.8 & -6.4 & -2.7 & 9 & 16 \\
      $y$-\text{coordinate} & 15 & -15 & 23.4 & 10.2 & 9.225 &
                                                                11.25
      & 15 \\ \hline
      \text{distance to} \ \ve x & 0 & 30 & 30 & 8 & 6.375 & 9.75 &
                                                                    16
      \\
      \text{distance to} \ \ve x' & 30 & 0 & 48 & 26 & 24.375 &
                                                                27.75
      & 34
    \end{array} 
  \]
  All values are exact and rational. The hyperbola is given by the
  equation $16 y^2 - 9 x^2 = 1296$ with eccentricity~$e = 5/3$, and
  the distance between the sheets is $d = 18$.

  Notice that the~$\ve a_i$ are in linear general position (points on
  one sheet of a hyperbola always are), and there are no
  mirror-symmetries between them. So we have a counterexample against
  the two-dimensional case of a recent conjecture by
  \mycite[page~5]{Hou:2022}. (Even selecting~$4$ of the~$5$
  points~$\ve a_i$ provides a counterexample.) 
\end{ex}

\begin{ex} \label{ex3D}%
  It is quite easy to turn \cref{exFive} into an example in dimension
  three. Just rotate the hyperbola around the major axis, and rotate
  the above points~$\ve a_i$ as well by some chosen angles. This is
  guaranteed to yield examples of nonuniqueness, but it takes a lucky
  choice of the rotation angles to obtain points that are in general
  linear position. We took the points $\ve a_1$, $\ve a_2$,
  and~$\ve a_5$ from the last example, and rotated each of them by two
  angles. This gives six points, which we combine with the upper
  vertex~$\ve a_0$ of $Q$. The result is given in the following table,
  which works like the previous one.
  \[
    \begin{array}{c|cc|ccccccc}
      \text{Point} & \ve x & \ve x' & \ve a_0 & \ve a_1' & \ve a_1'' &
      \ve a_2' & \ve a_2'' & \ve a_5' & \ve a_5'' \\ \hline
      $x$-\text{coordinate} & 0 & 0 & 0 & -28.8 & 0 & 6.4 & 0 & 9.6 &
                                                                      9.6 \\
      $y$-\text{coordinate} & 0 & 0 & 0 & 0 & -28.8 & 0 & -6.4 & -12.8
                             & 12.8 \\ 
      $z$-\text{coordinate} & 15 & -15 & 9 & 23.4 & 23.4 & 10.2 & 10.2
                  & 15 & 15 \\ \hline
      \text{distance to} \ \ve x & 0 & 30 & 6 & 30 & 30 & 8 & 8 & 16 &
                                                                       16 \\
      \text{distance to} \ \ve x' & 30 & 0 & 24 & 48 & 48 & 26 & 26 &
                                                                      34 & 34
    \end{array} 
  \]
  We used a computer to check that the seven
  points~$\ve a_0,\ve a_1' \upto \ve a_5''$ are in general linear
  position. We also checked that they have no symmetry in the
  following sense: there exists no nonidentity rigid transformation
  that permutes the seven points. This is true because, as it turns
  out, the distances between the points and their common center of
  gravity~$\ve c$ are pairwise different. But any rigid transformation
  permuting~$\ve a_0,\ve a_1' \upto \ve a_5''$ must map each point to
  a point that shares the same distance to~$\ve c$. We also verified
  this for subsets of size $\ge 5$ of the seven points, and found that
  with a single exception, all these subsets have no symmetries,
  either. For instance, the points
  $\ve a_0,\ve a_1',\ve a_1'',\ve a_2',\ve a_2''$ have no
  symmetry. Therefore they provide a three-dimensional counterexample to
  the conjecture of \mycite{Hou:2022} mentioned above. According to \mycite{Hou:2022}, \mycite{Abel:Chaffee:1991} postulated that there is a unique solution in the three-dimensional case if $m \ge 6$. Therefore we also have a counterexample to this postulate.

  We also obtain a counterexample to the following conjecture of \mycite[page~955]{Abel:Chaffee:1991}: ``It is conjectured that users on the `inside'
of a constellation will experience a unique fix, whereas
those `outside' of the constellation may not.'' In fact, it is visible from \cref{fFive} that in \cref{exFive} the user position~$\ve x$ belongs to the convex hull of the satellite positions~$\ve a_i$. Since the satellite positions in the present example were constructed by rotating those from \cref{exFive}, it should be clear that here, too, $\ve x$ lies in their convex hull. Explicitly, we have
\[
  \ve x = \frac{14}{54} \ve a_0 + \frac{10}{54} \ve a_1' + \frac{15}{54} \ve a_5' + \frac{15}{54} \ve a_5''.
\]
Thus~$\ve x$ lies in the convex hull even of $\ve a_o,\ve a_1',\ve a_5'$, and~$\ve a_5''$ (and therefore also in the convex hull of all seven points from this example). Since this example exhibits nonuniqueness, it provides a counterexample to \citename{Abel:Chaffee:1991}'s conjecture.
\end{ex}

\begin{ex} \label{exUnique}%
  In the case of $m = n+1$ satellites in dimension~$n$, the matrix $A$ never has rank~$n + 2$, so by
  \cref{tUnique} the question of uniqueness is determined by the type
  of the quadric $Q$. If the points~$\ve a_1 \upto \ve a_{n+1}$
  (standing for the satellite positions) are
  given, the type of $Q$ depends on the position~$\ve x$ of the
  user. \cref{fUnique} contains three two-dimensional examples of
  satellite positions and shows the areas of uniqueness and
  nonuniqueness.
  
  \begin{figure}[htbp]
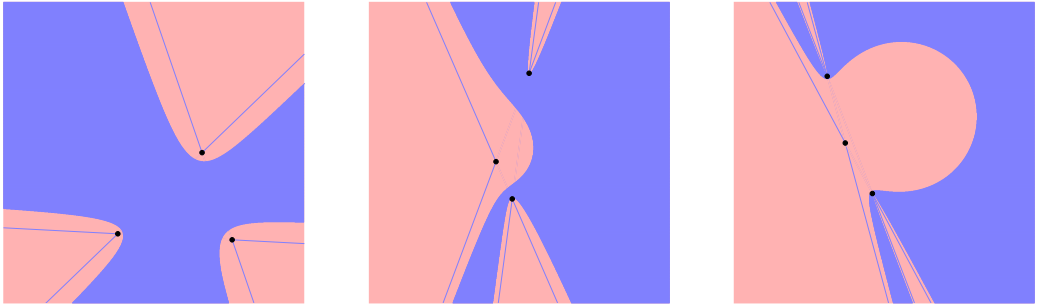

    \centering%
    \input{GPS_Unique1} \hspace{5mm} \input{GPS_Unique2}
    \hspace{5mm} \input{GPS_Unique3}
    \caption{In each picture, the black dots are the satellite
      positions. The blue areas and lines delineate the user positions
      where the global positioning problem has a unique solution.}
    \label{fUnique}
  \end{figure}
  
  Pictures of this kind have appeared in \mycite{Schmidt:1972} and
  \mycite{Abel:Chaffee:1991}. Those pictures do not indicate that the
  domain of uniqueness also contains lines, which appear in
  \cref{fUnique}. These lines arise as the locus where $Q$ is a cone
  (see \cref{tUnique}\ref{tUniqueE} and \cref{exCone}), so their
  absence in~[\citenumber{Schmidt:1972},
  \citenumber{Abel:Chaffee:1991}] is compatible with the fact that the
  case of a cone is missing from those papers.
\end{ex}

\begin{ex} \label{exDistribution}%
  The previous example is about the case of~$m = n+1$ satellites in dimension~$n$, and provides some pictures of regions where the global positioning problem has a unique solution. As an attempt to provide a more quantified picture of how often uniqueness occurs in this case, we ran a large number of experiments for $n = 2,3,4$. Each time, the satellite positions~$\ve a_i$ were chosen at random within a certain range. For a given configuration of satellites, we then tested uniqueness for many user positions~$\ve x$, also chosen randomly within the same range. We regard this as an approximation to the ``probability'' that the problem has a unique solution for a given satellite configuration. The results are shown in \cref{fDistribution}. For a probability~$p$ (a value on the $x$-axis), the histograms show for how many satellite configurations the probability of a unique solution is~$p$. The averages given in \cref{fDistribution} are taken over all combinations of a satellite configuration and user position that occurred. In total, we ran 600 million tests to produce the three histograms.
  
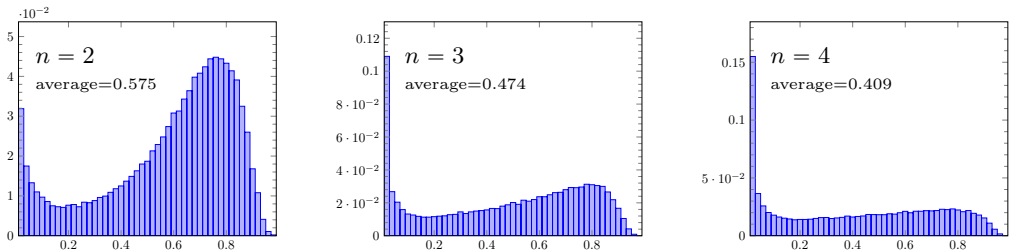
\begin{figure}[htbp]
\centering
\begin{tikzpicture}[scale=0.5]
\begin{axis}[
  ymin=0, ymax=0.0537,xmin=0.00999,xmax=0.990,
 minor y tick num = 4,
        area style,
]
\addplot+[ybar interval,mark=no] plot coordinates {
(0.00999,0.0319) (0.0300,0.0175) (0.0500,0.0133) (0.0699,0.0110) (0.0900,0.00969) (0.110,0.00859) (0.130,0.00752) (0.150,0.00728) (0.170,0.00712) (0.190,0.00767) (0.210,0.00784) (0.230,0.00728) (0.250,0.00842) (0.270,0.00829) (0.290,0.00883) (0.310,0.00961) (0.330,0.00995) (0.350,0.0109) (0.370,0.0118) (0.390,0.0125) (0.410,0.0137) (0.430,0.0149) (0.450,0.0163) (0.470,0.0180) (0.490,0.0187) (0.510,0.0213) (0.530,0.0229) (0.550,0.0248) (0.570,0.0274) (0.590,0.0308) (0.610,0.0313) (0.630,0.0343) (0.650,0.0364) (0.670,0.0398) (0.690,0.0408) (0.710,0.0424) (0.730,0.0444) (0.750,0.0448) (0.770,0.0444) (0.790,0.0433) (0.810,0.0414) (0.830,0.0391) (0.850,0.0325) (0.870,0.0260) (0.890,0.0168) (0.910,0.0108) (0.930,0.00412) (0.950,0.00107) (0.970,0.000210) (0.990,1.00E-5) };
\end{axis}
\draw (0.2,4.8) node[right] {\small $n = 2$};
\draw (0.2,4.0) node[right] {\tiny average=0.575};
\end{tikzpicture} \hspace{5mm} 
\begin{tikzpicture}[scale=0.5]
\begin{axis}[
 ymin=0, ymax=0.130,xmin=0.00999,xmax=0.990,
 minor y tick num = 4,
        area style,
]
\addplot+[ybar interval,mark=no] plot coordinates {
(0.00999,0.109) (0.0300,0.0268) (0.0500,0.0204) (0.0699,0.0159) (0.0900,0.0132) (0.110,0.0126) (0.130,0.0118) (0.150,0.0113) (0.170,0.0113) (0.190,0.0116) (0.210,0.0119) (0.230,0.0121) (0.250,0.0128) (0.270,0.0130) (0.290,0.0145) (0.310,0.0135) (0.330,0.0145) (0.350,0.0150) (0.370,0.0153) (0.390,0.0157) (0.410,0.0166) (0.430,0.0165) (0.450,0.0180) (0.470,0.0188) (0.490,0.0202) (0.510,0.0191) (0.530,0.0216) (0.550,0.0210) (0.570,0.0220) (0.590,0.0237) (0.610,0.0237) (0.630,0.0249) (0.650,0.0262) (0.670,0.0264) (0.690,0.0280) (0.710,0.0293) (0.730,0.0292) (0.750,0.0296) (0.770,0.0313) (0.790,0.0310) (0.810,0.0305) (0.830,0.0300) (0.850,0.0266) (0.870,0.0219) (0.890,0.0167) (0.910,0.0105) (0.930,0.00418) (0.950,0.000910) (0.970,3.00E-5) (0.990,0.000) };
\end{axis}
\draw (0.3,4.8) node[right] {\small $n = 3$};
\draw (0.3,4.0) node[right] {\tiny average=0.474};
\end{tikzpicture} \hspace{5mm}
\begin{tikzpicture}[scale=0.5]
\begin{axis}[
 ymin=0, ymax=0.185,xmin=0.00999,xmax=0.990,
 minor y tick num = 4,
        area style,
]
\addplot+[ybar interval,mark=no] plot coordinates {
(0.00999,0.155) (0.0300,0.0365) (0.0500,0.0258) (0.0699,0.0201) (0.0900,0.0178) (0.110,0.0162) (0.130,0.0152) (0.150,0.0145) (0.170,0.0138) (0.190,0.0141) (0.210,0.0141) (0.230,0.0143) (0.250,0.0150) (0.270,0.0157) (0.290,0.0158) (0.310,0.0149) (0.330,0.0154) (0.350,0.0159) (0.370,0.0170) (0.390,0.0162) (0.410,0.0167) (0.430,0.0171) (0.450,0.0184) (0.470,0.0183) (0.490,0.0182) (0.510,0.0182) (0.530,0.0190) (0.550,0.0187) (0.570,0.0199) (0.590,0.0211) (0.610,0.0199) (0.630,0.0212) (0.650,0.0212) (0.670,0.0217) (0.690,0.0215) (0.710,0.0218) (0.730,0.0229) (0.750,0.0226) (0.770,0.0232) (0.790,0.0223) (0.810,0.0208) (0.830,0.0218) (0.850,0.0194) (0.870,0.0179) (0.890,0.0154) (0.910,0.0109) (0.930,0.00565) (0.950,0.00160) (0.970,9.00E-5) (0.990,0.000) };
\end{axis}
\draw (0.3,4.8) node[right] {\small $n = 4$};
\draw (0.3,4.0) node[right] {\tiny average=0.409};
\end{tikzpicture}
\caption{Histograms of fraction of user position~$\ve x$ with unique solution for a random arrangement~$\ve a_i$ of $n+1$ satellites, for dimensions~$n = 2,3,4$.}
\label{fDistribution}
\end{figure}

We were surprised by the shape of the histograms in \cref{fDistribution}. It may be insightful to investigate the reasons why they look that way.
\end{ex}

\section{Uniqueness for almost all positions $\ve x$}
\label{sAlmost}

It is a consequence of \cref{tUnique} that if the number of sensors is
$m = n + 1$, then the basic equation~\eqref{eqBasic} does not have a
unique solution for ``many'' positions $\ve x \in \RR^n$ (see \cref{exUnique,exDistribution}).
In this section we look at the case $m \ge n +
2$. The following result says that then indeed we have uniqueness for
almost all~$\ve x$.

\begin{theorem}[Uniqueness for almost all~$\ve x$] \label{tM5}%
  Assume that in the situation of \cref{Setup} we have $m \ge n +
  2$. Then~\eqref{eqAbsolute}, and therefore also~\eqref{eqBasic}, has
  a unique solution for almost all~$\ve x \in \RR^n$.

  More precisely, there exists a nonzero polynomial~$f$ in~$n$
  variables such that the matrix $A \in \RR^{m \times (n + 2)}$,
  defined by~\eqref{eqA}, has rank~$n + 2$ for all~$\ve x$ with
  $f(\ve x) \ne 0$. (Recall that $A$ depends on~$\ve x$
  by~\eqref{eqSetup}.) So~\eqref{eqTxA} gives the unique solution
  of~\eqref{eqAbsolute}, and therefore of~\eqref{eqBasic}, for all
  those~$\ve x$.
\end{theorem}

\begin{rem*}
  Apart from $m \ge n + 2$, the hypotheses on the points~$\ve a_i$
  made in the theorem (by way of \cref{Setup}) are that they are
  pairwise distinct and do not lie on a common affine
  hyperplane. These hypotheses are clearly also necessary, since equal
  points~$\ve a_i$ add no further equations to the
  system~\eqref{eqBasic}, and since for points on an affine
  hyperplane, \eqref{eqBasic} almost never has a unique solution (see
  \cref{rLateration}\ref{rLaterationB}). In this way, \cref{tM5} is
  the best result that can possibly be expected. It may be remarkable
  that the theorem does {\em not} require that the points~$\ve a_i$
  are in general linear position.
\end{rem*}

As we will see in a moment, the following result implies \cref{tM5}.

\begin{prop} \label{pM5}%
  Let $\ve a_1 \upto \ve a_{n+2} \in \RR^n$, with $n \ge 2$, be
  pairwise distinct points that do not lie on a common affine
  hyperplane. Then there exists a nonzero polynomial~$f$ such that
  $f(\ve x) \ne 0$ for an~$\ve x \in \RR^n$ implies
  \begin{equation} \label{eqAPM}%
    \det
    \begin{pmatrix}
      \varepsilon_1 \cdot \lVert\ve a_1 - \ve x\rVert & \ve a_1^T & 1 \\
      \vdots & \vdots & \vdots \\
      \varepsilon_{n+2} \cdot \lVert\ve a_{n+2} - \ve x\rVert & \ve
      a_{n+2}^T & 1
    \end{pmatrix} \ne 0
  \end{equation}
  for all $\varepsilon_1 \upto \varepsilon_{n+2} = \pm 1$.
\end{prop}

\begin{proof}[Proof of the implication ``\cref{pM5} $\Rightarrow$
  \cref{tM5}'']
  In \cref{tM5} we are given satellite positions $\ve a_1 \upto \ve a_m \in \RR^n$
  distinct and not contained in an affine hyperplane. We are also
  given~$t$ by \cref{Setup}, and for~$\ve x \in \RR^n$, $t_i$ is set
  to $t_i := \lVert\ve a_i - \ve x\rVert + t$. Now subtracting any
  real number from~$t$ and, consequently, from the~$t_i$ does not
  change the number of solutions of~\eqref{eqBasic}
  and~\eqref{eqAbsolute}. Neither does it change the rank of $A$,
  since subtracting the same number from all~$t_i$ amounts to
  subtracting a multiple of the last column of $A$ from the first
  column of $A$. So we may subtract~$t$ itself and thus assume
  $t_i = \lVert\ve a_i - \ve x\rVert$.

  By renumbering the~$\ve a_i$ we may also assume that the first
  $n + 2$ of them do not lie on a common affine hyperplane. These can
  be ``fed into'' \cref{pM5}, which gives us a nonzero
  polynomial~$f$. If $f(\ve x) \ne 0$, then the matrices
  in~\eqref{eqAPM} all have rank~$n + 2$, and in particular the one
  with $\varepsilon_1 = \cdots = \varepsilon_{n+2} = 1$. This clearly
  implies $\rank(A) = n + 2$. (In fact, the~$\varepsilon_i$ introduced
  in \cref{pM5} are irrelevant for deriving \cref{tM5}. But they
  provide some benefit, as explained in \cref{rM5}, and come at no
  extra cost for the proof.)
\end{proof}

Before proving \cref{pM5}, let us describe the rough idea behind the
proof in an informal, inexact way. First, the hypothesis that the
determinant in~\eqref{eqAPM} is zero for some~$\varepsilon_i$ is
readily translated to the condition that the $\ve a_i$ lie on a common
quadric of revolution with focus~$\ve x$. Now by
\mycite{Gfrerrer:Zsombor:2009}, a quadric of revolution $Q$ is
uniquely determined by the direction of its symmetry axis and $n + 2$
points in general position on $Q$. (In
fact,~[\citenumber{Gfrerrer:Zsombor:2009}] has this result for
$n = 3$.) In our situation, the $n + 2$ points are given, so the only
free parameter for defining $Q$ is the direction of the symmetry axis,
which can be viewed as a point in projective space $\PP^{n - 1}$,
which has dimension~$n-1$. But each quadric of revolution has at most
two foci. Therefore, since the quadrics in question have $\PP^{n-1}$
as ``moduli space'', the possible foci can also only range through a
variety of that dimension. So a point~$\ve x$ not
satisfying~\eqref{eqAPM} must lie in a subvariety of $\RR^n$ of
smaller dimension.

However, turning this idea into a proof of \cref{pM5} presents several
obstacles:

\begin{itemize}
\item \citename{Gfrerrer:Zsombor:2009}'s result must be generalized
  from dimension~$3$ to dimension~$n$.
\item The idea that if objects with an $(n - 1)$-dimensional moduli
  space each yield at most two points~$\ve x$, then these points lie
  in a variety of dimension~$n - 1$ must be turned into a rigorous
  proof.
\item The biggest obstacle is that
  in~[\citenumber{Gfrerrer:Zsombor:2009}], points in ``general
  position'' are assumed to be non-cospherical, and this is used
  heavily in the proof. Our \cref{pM5} does not make this assumption,
  and it is especially unpalatable since in the global positioning situation the
  satellites do fly on a common sphere. So we need to get along
  without this hypothesis. On the other hand, we do not need the full
  strength of \citename{Gfrerrer:Zsombor:2009}'s result.
\end{itemize}

As it turns out, our actual proof looks nothing like the above
idea. We feel that just because of this, our rough sketch may be all
the more helpful for readers. However, in contrast to the proof
sketch, the actual proof requires methods from algebraic
geometry.

\begin{proof}[Proof of \cref{pM5}]
  We set $m := n+2$. Take a point~$\ve x \in \RR^n$ such that there
  are $\varepsilon_i \in \{1,-1\}$ making the
  determinant~\eqref{eqAPM} zero. Write
  $A_{\varepsilon_1 \upto \varepsilon_m}$ for the matrix in question,
  which thus has rank $< m$. The hypothesis that the~$\ve a_i$ do not
  lie on an affine hyperplane means that the matrix obtained by
  removing the first column from
  $A_{\varepsilon_1 \upto \varepsilon_m}$ has rank~$m-1$. So the first
  column is a linear combination of the other columns, which we write
  as
  \[
    \begin{pmatrix}
      \varepsilon_1 \lVert\ve a_1 - \ve x\rVert \\
      \vdots \\
      \varepsilon_m \lVert\ve a_m - \ve x\rVert
    \end{pmatrix}
    =
    \begin{pmatrix}
      \ve a_1^T & 1 \\
      \vdots & \vdots \\
      \ve a_m^T & 1
    \end{pmatrix} \cdot
    \begin{pmatrix}
      \ve u \\
      \alpha
    \end{pmatrix}
  \]
  with $\ve u \in \RR^n$ and~$\alpha \in \RR$. Notice that~$\ve u$
  and~$\alpha$ depend on~$\ve x$. Picking out the entries of this
  equation and squaring both sides yields
  \[
    \lVert\ve a_i\rVert^2 - 2\langle\ve a_i,\ve x\rangle +
    \lVert\ve x\rVert^2 = \langle\ve a_i,\ve u\rangle^2 + 2 \alpha
    \langle\ve a_i,\ve u\rangle + \alpha^2 \quad (i = 1 \upto m)
  \]
  Writing
  \[
    \ve a_i =
    \begin{pmatrix}
      a_{i,1} \\ \vdots \\ a_{i,n}
    \end{pmatrix}, \ \ve u =
    \begin{pmatrix}
      u_1 \\ \vdots \\ u_n
    \end{pmatrix} \quad \text{and} \quad \ve x =
    \begin{pmatrix}
      x_1 \\ \vdots \\ x_n
    \end{pmatrix},
  \]
  we get, for $i = 1 \upto m$,
  \begin{equation} \label{eqAUXAlpha}%
    \sum_{j=1}^n a_{i,j}^2 (u_j^2 - 1) + \sum_{1 \le j < k \le n} 2
    a_{i,j} a_{i,k} u_j u_k + \sum_{j=1}^n 2 a_{i,j} (x_j + \alpha
    u_j) + \alpha^2 - \lVert\ve x\rVert^2 = 0,
  \end{equation}
  or in matrix form
  \[
    \begin{pmatrix}
      1 & 2 a_{1,1} & \cdots & 2 a_{1,n} & a_{1,1}^2 & \cdots &
      a_{1,n}^2 & 2 a_{1,1} a_{1,2} & \cdots & 2 a_{1,n-1} a_{1,n} \\
      1 & 2 a_{2,1} & \cdots & 2 a_{2,n} & a_{2,1}^2 & \cdots &
      a_{2,n}^2 & 2 a_{2,1} a_{2,2} & \cdots & 2 a_{2,n-1} a_{2,n} \\
      \vdots & \vdots & & \vdots & \vdots & & \vdots & \vdots & &
      \vdots \\
      1 & 2 a_{i,1} & \cdots & 2 a_{i,n} & a_{i,1}^2 & \cdots &
      a_{i,n}^2 & 2 a_{i,1} a_{i,2} & \cdots & 2 a_{i,n-1} a_{i,n} \\
      \vdots & \vdots & & \vdots & \vdots & & \vdots & \vdots & &
      \vdots \\
      1 & 2 a_{m,1} & \cdots & 2 a_{m,n} & a_{m,1}^2 & \cdots &
      a_{m,n}^2 & 2 a_{m,1} a_{m,2} & \cdots & 2 a_{m,n-1} a_{m,n}
    \end{pmatrix}
    \begin{pmatrix}
      \alpha^2 - \lVert\ve x\rVert^2 \\
      x_1 + \alpha u_1 \\
      \vdots \\
      x_n + \alpha u_n \\
      u_1^2 - 1 \\
      \vdots \\
      u_n^2 - 1 \\
      u_1 u_2 \\
      \vdots \\
      u_{n-1} u_n
    \end{pmatrix} = \ve 0.
  \]
  By \cref{lM5}, which is stated and proved below, the left-hand
  matrix has rank~$n + 2 = m$, and by the hypothesis that
  the~$\ve a_i$ do not lie on an affine hyperplane, the
  $(n+2) \times (n+1)$-submatrix formed by the leftmost $n + 1$
  columns has rank $n + 1$. Therefore by Gaussian elimination the
  left-hand matrix in the above equation can be brought to the form
  \[
    \left(
      \begin{array}{ccc|ccccccc}
        & & \\
        & I_{n+1} & & & & & \Asterisk \\
        & & \\
        \hline
        0 & \cdots & 0 & \gamma_{1,1} & \cdots & \gamma_{n,n} &
                                                                \gamma_{1,2}
        & \gamma_{1,3} & \cdots & \gamma_{n-1,n}
      \end{array}
    \right),
  \]
  where the $\gamma_{i,j} \in \RR$ are not all zero. Notice that this
  matrix only depends on the given points~$\ve a_i$ and not
  on~$\ve x$. From the first $n + 1$ rows of the above matrix equation
  we obtain
  \begin{align}
    \alpha^2 - \lVert\ve x\rVert^2 + h(u_1 \upto u_n) & =
                                                        0, \label{eqAX} \\
    x_i + \alpha u_i + h_i(u_1 \upto u_n) & = 0 \quad (i = 1 \upto
                                            n), \label{eqXU}
  \end{align}
  where $h,h_i \in \RR[U_1 \upto U_n]$ are polynomials in
  indeterminates $U_i$, and from the last row we get
  \begin{equation} \label{eqG}%
    g(u_1 \upto u_n) = 0 \quad \text{with} \quad g := \sum_{i=1}^n
    \gamma_{i,i} (U_i^2 - 1) + \!\!\! \sum_{1 \le i < j \le n} \!\!\!
    \gamma_{i,j} U_i U_j \ne 0.
  \end{equation}
  The polynomials~$g$, $h$ and~$h_i$ only depend on the
  points~$\ve a_i$ and not on~$\ve x$. From~\eqref{eqXU} we obtain
  \[
    \lVert\ve x\rVert^2 = \sum_{i=1}^n x_i^2 = \sum_{i=1}^n
    \bigl(h_i(u_1 \upto u_n)^2 + 2 u_i h_i(u_1 \upto u_n) \alpha +
    u_i^2 \alpha^2\bigr),
  \]
  and substituting this into~\eqref{eqAX} yields
  \[
    \left(\sum_{i=1}^n u_i^2 - 1\right) \alpha^2 + 2
    \left(\sum_{i=1}^n u_i h_i(u_1 \upto u_n)\right) \alpha +
    \sum_{i=1}^n h_i(u_1 \upto u_n)^2 - h(u_1 \upto u_n) = 0.
  \]
  We can write this as
  \begin{equation} \label{eqGUA}%
    g_1(u_1 \upto u_n) \alpha^2 + g_2(u_1 \upto u_n) \alpha + g_3(u_1
    \upto u_n) = 0
  \end{equation}
  with $g_1,g_2,g_3 \in \RR[U_1 \upto U_n]$, and specifically
  \begin{equation} \label{eqG1}%
    g_1 = \sum_{i=1}^n U_i^2 - 1.
  \end{equation}
  Let us summarize our findings so far. For every
  point~$\ve x \in \RR^n$ that makes at least one of the determinants
  in~\eqref{eqAPM} vanish, there exist~$\alpha \in \RR$ and
  $\ve u \in \RR^n$ satisfying~\eqref{eqXU}, \eqref{eqG},
  and~\eqref{eqGUA}. The polynomials $h_i,g,g_1,g_2$,
  and~$g_3 \in \RR[U_1 \upto U_n]$ occurring in these equations only
  depend on the given points~$\ve a_i$ and not on~$\ve x$.

  We now take an additional indeterminate $V$ and consider the ideal
  \[
    I := \bigl(g,g_1 V^2 + g_2 V + g3\bigr) \subseteq \CC[U_1 \upto
    U_n,V]
  \]
  and the affine variety $X \subseteq \CC^{n+1}$ given by $I$. With
  $Y \subseteq \CC^n$ given by the polynomial~$g$, the projection to
  the first~$n$ coordinates is a morphism $\map{\pi}{X}{Y}$. For a
  point $y \in Y$ the fiber $\pi^{-1}(\{y\})$ is isomorphic to the
  variety consisting of all $\alpha \in \CC$ with
  $g_1(y) \alpha^2 + g_2(y) \alpha + g_2(y) = 0$. Let $X' \subseteq X$
  be a closed, irreducible subset. Then the Zariski closure
  $Y' := \overline{\pi(X')} \subseteq Y$ of the image is also
  irreducible. We distinguish two cases.
  \begin{description}
  \item[Case 1] there exists $y \in \pi(X')$ such that $g_1(y) \ne
    0$. Then the fiber $\pi^{-1}(\{y\})$ is finite and therefore
    $0$-dimensional. By Chevalley's theorem on fiber dimension (see
    \mycite[Corollary~10.6]{Kemper.Comalg} or \mycite[Chapter~II,
    Exercise~3.22(b)]{hart}), we obtain
    \[
      0 \ge \dim(X') - \dim(Y').
    \]
    But $\dim(Y') \le \dim(Y) = n-1$ since $g \ne 0$, so
    $\dim(X') \le n-1$.
  \item[Case 2] $g_1(y) = 0$ for every $y \in \pi(X')$. Then~$g_1$
    vanishes on $\pi(X')$ and therefore also on the closure $Y'$,
    which means that both~$g$ and~$g_1$ lie in the vanishing ideal of
    $Y'$. Now by looking at its definition~\eqref{eqG1}, we see that
    $g_1 \in \CC[U_1 \upto U_n]$ is irreducible (since $n \ge
    2$). Moreover~\eqref{eqG} implies that~$g_1$ does not divide~$g$
    (again since $n \ge 2$). So $(g_1) \subset \CC[U_1 \upto U_n]$ is
    a nonzero prime ideal strictly contained in $(g_1,g)$, and it
    follows that the Krull dimension of the quotient ring
    $\CC[U_1 \upto U_n]/(g_1,g)$ is at most $n - 2$. Therefore
    \[
      \dim(Y') \le n-2.
    \]
    Take $y \in \pi(X')$ then the fiber $\pi^{-1}(\{y\})$ has
    dimension~$0$ or~$1$ (the latter only occurs if
    $g_2(y) = g_3(y) = 0$, but this cannot be ruled out). Using
    Chevalley's theorem again, we now obtain
    $\dim(X') \le \dim(Y') + 1 \le n - 1$.
  \end{description}
  Since in both cases $\dim(X') \le n-1$, we conclude that also $X$
  has dimension $\le n-1$, and therefore the quotient ring
  $\CC[U_1 \upto U_n,V]/I$ has Krull dimension $\le n-1$. But for
  finitely generated $K$-algebras the Krull dimension equals the
  transcendence degree (see~[\citenumber{Kemper.Comalg}, Theorem~5.9]
  or~[\citenumber{hart}, Chapter~I, Theorem1.8A]), so~$n$ elements
  from this algebra are always algebraically dependent. In particular,
  this applies to the classes modulo $I$ of the polynomials
  $-h_i - V U_i$ ($i = 1 \upto n$). So there exists a nonzero
  polynomial $f \in \CC[T_1 \upto T_n]$ (with $T_i$ new
  indeterminates) such that
  $f\bigl(-h_1 - V U_1 \upto -h_n - V U_n\bigr) \in I$, so we can
  write
  \begin{equation} \label{eqFHVU}%
    f\bigl(-h_1 - V U_1 \upto -h_n - V U_n\bigr) = f_1 \cdot g + f_2
    \cdot \bigl(g_1 V^2 + g_2 V + g_3\bigr)
  \end{equation}
  with $f_1,f_2 \in \CC[U_1 \upto U_n,V]$. We claim that
  $f(\ve x) = 0$ for every point $\ve x \in \RR^n$ that makes at least
  one of the determinants in~\eqref{eqAPM} vanish. Indeed, as we have
  seen, for such an~$\ve x = (x_1 \upto x_n)$ there
  exist~$\alpha \in \RR$ and $\ve u = (u_1 \upto u_n) \in \RR^n$
  satisfying~\eqref{eqXU}, \eqref{eqG}, and~\eqref{eqGUA}, so
  \begin{multline*}
    f(x_1 \upto x_n) \underset{\eqref{eqXU}}{=} f\bigl(-h_1(\ve u) -
    \alpha u_1 \upto -h_n(\ve u) - \alpha u_n\bigr)
    \underset{\eqref{eqFHVU}}{=} \\
    f_1(\ve u,\alpha) \cdot g(\ve u) + f_2(\ve u,\alpha) \cdot
    \bigl(g_1(\ve u) \alpha^2 + g_2(\ve u) \alpha + g_3(\ve u)\bigr)
    \underset{\eqref{eqG}, \eqref{eqGUA}}{=} 0.
  \end{multline*}
  With this \cref{pM5} is almost proved. The one thing that is missing
  is that our polynomial~$f$ has complex coefficients, whereas the
  polynomial whose existence was asserted in the proposition was of
  course understood to be real. But since the~$x_i$ in the above
  equation are all real, we can extract the real and imaginary part
  of~$f$ coefficient-wise, and obtain
  $\operatorname{Re}(f)(\ve x) = 0$ and
  $\operatorname{Im}(f)(\ve x) = 0$. Since at least one of the
  polynomials $\operatorname{Re}(f)$ or $\operatorname{Im}(f)$ is
  nonzero, it provides the desired polynomial for the proposition.
\end{proof}

The following lemma was used in the proof.

\begin{lemma} \label{lM5}%
  In the situation of \cref{pM5}, let $t_1 \upto t_k$, with
  $k = \binom{n+2}{2}$, be all the monomials in variables $y_1 \upto y_n$
  of degree $\le 2$. Then the matrix
  \[
    M(\ve a_1 \upto \ve a_{n+2}) := \Bigl(t_j(\ve
      a_i)\Bigr)_{\substack{i = 1 \upto n+2 \\ j = 1 \upto k}} \in
    \RR^{(n+2) \times k}
  \]
  has rank $n + 2$.
\end{lemma}

\begin{proof}
  By hypothesis the $\ve a_i$ do not lie on an affine hyperplane,
  i.e., the matrix
  $\left(\begin{smallmatrix} \ve a_1 & \cdots & \ve a_{n+2} \\ 1 &
      \cdots & 1\end{smallmatrix}\right) \in \RR^{(n+1) \times (n+2)}$
  has rank~$n+1$. By renumbering the~$\ve a_i$ we may assume that the
  first $n+1$ columns are linearly independent, which means that
  $\ve a_1 \upto \ve a_{n+1}$ do not lie on an affine
  hyperplane. Write $\ve a_{n+2} = (c_1 \upto c_n)^T$, and for a
  monomial $t = y_1^{e_1} \cdots y_n^{e_n}$ write
  $\tilde t := \prod_{i=1}^n (y_i - c_i)^{e_i}$. This is equal to~$t$
  plus a linear combination of lower-degree monomials, so we have a
  matrix $C \in \RR^{k \times k}$ such that
  \[
    \bigl(\tilde t_1 \upto \tilde t_k\bigr) = \bigl(t_1 \upto t_k)
    \cdot C.
  \]
  If we order the~$t_i$ by degrees, $C$ is a triangular matrix with
  $1$'s on the diagonal, so $C \in \GL_k(\RR)$. By setting
  $\tilde {\ve a}_i := \ve a_i - \ve a_{n+2}$ we obtain
  $\tilde t_j(\ve a_i) = t_j(\tilde{\ve a}_i)$, so
  $M(\tilde{\ve a}_1 \upto \tilde{\ve a}_{n+2}) = M(\ve a_1 \upto \ve
  a_{n+2}) \cdot C$ and therefore
  \[
    \rank\bigl(M(\ve a_1 \upto \ve a_{n+2})\bigr) =
    \rank\bigl(M(\tilde{\ve a}_1 \upto \tilde{\ve a}_{n+2})\bigr).
  \]
  Since $\tilde{\ve a}_{n+2} = \ve 0$ and $t_1 = 1$ we have
  \[
    M(\tilde{\ve a}_1 \upto \tilde{\ve a}_{n+2}) = \left(
      \begin{array}{c|ccc}
        1 & & & \\
        \vdots & & M^+(\tilde{\ve a}_1 \upto \tilde{\ve a}_{n+1}) & \\
        1 & & & \\
        \hline
        1 & 0 & \cdots & 0
      \end{array}
    \right),
  \]
  where
  $M^+(\tilde{\ve a}_1 \upto \tilde{\ve a}_{n+1}) :=
  \Bigl(t_j(\tilde{\ve a}_i)\Bigr)_{\substack{i = 1 \upto n+1 \\ j = 2
      \upto k}} \in \RR^{(n+1) \times (k-1)}$ is formed by evaluating
  the monomials of positive degree. So we need to show that
  $\rank\bigl(M^+(\tilde{\ve a}_1 \upto \tilde{\ve a}_{n+1})\bigr) =
  n+1$. Observe that $\tilde{\ve a}_i \ne \ve 0$ ($1 \le i \le n+1$)
  since the $\ve a_i$ are pairwise distinct by hypothesis, and that
  $\tilde{\ve a}_1 \upto \tilde{\ve a}_{n+1}$ do not lie on an affine
  hyperplane. Therefore the following claim implies the desired
  rank-equation.
  \begin{claim*}
    Let $\ve b_1 \upto \ve b_m \in \RR^n \setminus \{\ve 0\}$ be
    affinely independent (i.e., the matrix
    $\left(\begin{smallmatrix} \ve b_1 & \cdots & \ve b_m \\ 1 &
        \cdots & 1\end{smallmatrix}\right) \in \RR^{(n+1) \times m}$
    has rank~$m$). Then
    $M^+(\ve b_1 \upto \ve b_m) \in \RR^{m \times (\binom{n+2}{2}-1)}$
    also has rank~$m$.
  \end{claim*}
  We prove the claim by induction on~$n$. Let $d_1 \upto d_m$ be the
  first coordinates of the vectors $\ve b_1 \upto \ve b_m$, and
  renumber the~$\ve b_i$ such that $d_i \ne 0$ for $1 \le i \le r$ and
  $d_{r+1} = \cdots = d_m = 0$. We have $0 \le r \le m$, and $r = m$
  if $n = 1$. At this point, it is helpful to reorder the monomials
  $t_2 \upto t_k$ (which are precisely the monomials of degree~$1$
  or~$2$) such that the monomials
  $y_1,y_1 ^2, y_1 y_2,y_1 y_3 \upto y_1 y_n$ come first, and are
  followed by the monomials of degree~$1$ or~$2$ in the
  variables~$y_2 \upto y_n$. Then
  \begin{equation} \label{eqMBDE}%
    M^+(\ve b_1 \upto \ve b_m) = \left(
      \begin{array}{c|c}
        D & \Asterisk \\
        \hline
        0 & E
      \end{array}\right),
  \end{equation}
  where
  \[
    D = 
    \begin{pmatrix}
      d_1 \\
      & \ddots \\
      & & d_r
    \end{pmatrix} \cdot
    \begin{pmatrix}
      1 & \ve b_1 ^T \\
      \vdots & \vdots \\
      1 & \ve b_r^T
    \end{pmatrix} \in \RR^{r \times (n+1)},
  \]
  and where $E$ is formed by evaluating all monomials of degree~$1$
  or~$2$ in the variables~$y_2 \upto y_n$ at
  $\ve b_{r+1} \upto \ve b_m$. It follows from the affine independence
  of $\ve b_1 \upto \ve b_r$ that $\rank(D) = r$, so the rows of $D$
  are linearly independent. If $r = m$, which is guaranteed in the
  case $n = 1$, this finishes the proof. On the other hand, if
  $r < m$, then with $\map{\pi}{\RR^n}{\RR^{(n-1)}}$ the projection on
  the last $n-1$ coordinates, we obtain
  \[
    E = M^+\bigl(\pi(\ve b_{r+1}) \upto \pi(\ve b_m)\bigr) \in
    \RR^{(m-r) \times (\binom{n+1}{2}-1)}.
  \]
  Since $\ve b_{r+1} \upto \ve b_m$ are affinely independent and their
  first coordinate is zero, their images under~$\pi$ are also affinely
  independent. The $\pi(\ve b_i)$ are also nonzero for $i > r$ since
  $\ve b_i \ne \ve 0$. So by induction we get $\rank(E) = m - r$, so
  the rows of $E$ are linearly independent. Now from~\eqref{eqMBDE} we
  see that the rank of $M^+(\ve b_1 \upto \ve b_m)$ is as claimed.
\end{proof}

\begin{rem} \label{rM5}%
  The proof of \cref{pM5} is not constructive. But, intriguingly, the
  proposition gives rise to its own constructive version, and with
  that to a constructive version of \cref{tM5}. In fact, we claim that
  \begin{equation} \label{eqFtilde}%
    \tilde f(\ve x) := \prod_{\varepsilon_2 \upto \varepsilon_m \in \{\pm
      1\}} \det\begin{pmatrix}
      \lVert\ve a_1 - \ve x\rVert & \ve a_1^T & 1
      \\
      \varepsilon_2 \cdot \lVert\ve a_2 - \ve x\rVert & \ve a_2^T & 1 \\
      \vdots & \vdots & \vdots \\
      \varepsilon_m \cdot \lVert\ve a_m - \ve x\rVert & \ve
      a_m^T & 1
    \end{pmatrix}
  \end{equation}
  (with $m := n+2$) is a polynomial in the coordinates~$x_i$
  of~$\ve x$, and can be taken as the polynomial~$f$ whose existence
  is asserted in \cref{pM5}. (One can also include an~$\varepsilon_1$
  as a sign before $\lVert\ve a_1 - \ve x\rVert$ in the
  multiplication, but that makes~$\tilde f$ unnecessarily large.)

  In fact, let $d(y_1 \upto y_m)$ be the
  determinant of a given $m \times m$ matrix with indeterminates
  $y_1 \upto y_m$ in the first column. Then
  $g := \prod_{\varepsilon_2 \upto \varepsilon_m \in \{\pm 1\}}
  d(y_1,\varepsilon_2 y_2 \upto \varepsilon_m y_m)$ is invariant under
  the automorphisms sending~$y_i$ to~$-y_i$ ($i = 2 \upto m$). But
  also
  \begin{multline*}
    g(-y_1,y_2 \upto y_m) = \prod_{\varepsilon_2 \upto \varepsilon_m
      \in \{\pm 1\}} d(-y_1,-\varepsilon_2 y_2 \upto -\varepsilon_m
    y_m) = \\
    \prod_{\varepsilon_2 \upto \varepsilon_m \in \{\pm 1\}}
    \bigl(-d(y_1,\varepsilon_2 y_2 \upto \varepsilon_m y_m)\bigr) =
    g(y_1,y_2 \upto y_m).
  \end{multline*}
  So only the squares of the~$y_i$ occur in~$g$. Hence by specializing
  the $y_i$ to $\lVert\ve a_i - \ve x\rVert$, we see
  that~$\tilde f(\ve x)$ is indeed a polynomial in the~$x_i$. Now
  \cref{pM5} tells us that there exists an~$\ve x$ for which none of
  the determinants in~\eqref{eqFtilde} vanishes. This implies
  $\tilde f \ne 0$. Moreover, if $\tilde f(\ve x) \ne 0$, then the
  determinant in~\eqref{eqAPM} does not vanish
  for~$\varepsilon_1 = 1$, and by the above calculation also
  for~$\varepsilon_1 = -1$.
\end{rem}

\section{Uniqueness for all positions~$\ve x$} \label{sAllX}

In this section, as in the previous one, we consider given
positions~$\ve a_i$ of satellites. The question then is for which positions~$\ve x$ of the
GPS-user the basic system of
equations~\eqref{eqBasic} has a unique solution. In the previous
section the answer was that if $m \ge n+2$, we have uniqueness for
{\em almost} all~$\ve x$. It would, of course be better to have
uniqueness for all~$\ve x$, without exceptions. This is not only of theoretical interest, since in the
vicinity of points where a problem is ill-posed it tends behave badly
in the presence of inexact measurements and other errors; so if it is
well-posed everywhere, this problem does not arise.


\begin{theorem}[Uniqueness for all~$\ve x$] \label{tM8}%
  Assume that in the situation of \cref{Setup} we have
  $m \ge 2 n + 2$. Then for almost all~$\ve a_1 \upto \ve a_m$ we have
  that~\eqref{eqAbsolute}, and therefore also~\eqref{eqBasic}, has a
  unique solution for all~$\ve x \in \RR^n$.

  More precisely, there exists a nonzero polynomial~$f$ in~$m n$
  variables such that if $f(\ve a_1 \upto \ve a_m) \ne 0$, then for
  all~$\ve x \in \RR^n$ the matrix $A \in \RR^{m \times (n + 2)}$,
  defined by~\eqref{eqA}, has rank~$n + 2$. (Recall that $A$ depends
  on~$\ve x$ by~\eqref{eqSetup}.) So~\eqref{eqTxA} gives the unique
  solution of~\eqref{eqAbsolute}, and therefore of~\eqref{eqBasic}.
\end{theorem}

Again, the proof uses some algebraic geometry.

\begin{proof}
  Consider the affine variety $X \subseteq \CC^{m n + 2 n + 1}$ given
  by the equations
  \begin{equation} \label{eqUAVX} \sum_{j=1}^n (U_j^2 - 1) A_{i,j} + 2
    \!\!\! \sum_{1 \le j < k \le n} \!\! U_j U_k A_{i,j} A_{i,k} + 2
    \sum_{j=1}^n (X_j - V U_j) A_{i,j} + V^2 - \sum_{j=1}^n X_j^2
    \quad (i = 1 \upto m)
  \end{equation}
  with indeterminates $A_{i,j}$, $U_j$, $X_j$, and~$V$. Mapping a
  point of $X$ to the coordinates corresponding to the indeterminates
  $U_j$, $X_j$, and $V$ gives a morphism
  $\map{\pi}{X}{\CC^{2 n + 1}}$. Every fiber of~$\pi$ is isomorphic to
  the (cartesian) product of~$m$ copies of a hypersurface in $\CC^n$,
  so its dimension is $m (n-1)$. Hence as in the proof of \cref{pM5}
  (but without any case distinction) we see that
  $\dim(X') - \dim(Y') \le m (n-1)$ for every closed, irreducible
  subset $X' \subseteq X$, with $Y' := \overline{\pi(X')}$. Since
  $\dim(Y') \le 2 n + 1$, we conclude that
  $\dim(X) \le m (n-1) + 2 n + 1$. Since $m > 2 n + 1$, this implies
  \begin{equation} \label{eqDimX}%
    \dim(X) < m n.
  \end{equation}
  Now consider the morphism $\map{\psi}{X}{\CC^{m n}}$ given by
  mapping a point of $X$ to the coordinates corresponding to the
  indeterminates $A_{i,j}$. It follows from~\eqref{eqDimX} that the
  Zariski closure $Y := \overline{\psi(X)}$ of the image has dimension
  $< m n$, so there is a nonzero $f \in \CC[A_{1,1} \upto A_{m,n}]$
  that vanishes on $Y$.
  
  To finish the proof, assume that for
  $\ve a_1 \upto \ve a_m \in \RR^n$ there exists~$\ve x$ such that the
  matrix $A$ has rank $< r + 2$. We have seen that this rank does not
  change if we subtract~$t$ from the~$t_i$, i.e., if we
  assume~$t = 0$. Then \cref{lSide} tells us that
  $\langle\ve u,\ve a_i\rangle - \alpha = \lVert\ve a_i - \ve x\rVert$
  for all~$i$, and squaring both sides yields
  \begin{equation} \label{eqUAX}%
    \langle\ve u,\ve a_i\rangle^2 - \langle\ve a_i,\ve a_i\rangle + 2
    \langle \ve x - \alpha \ve u,\ve a_i\rangle + \alpha^2 -
    \langle\ve x,\ve x\rangle = 0 \quad (i = 1 \upto m).
  \end{equation}
  This means that~$\alpha$ and the coordinates of $\ve u$, $\ve x$ and
  the~$\ve a_i$ satisfy the equations~\eqref{eqUAVX} defining $X$. So
  the point $(\ve a_1 \upto \ve a_m) \in \RR^{m n}$ lies in the image
  $\psi(X)$ and therefore in $Y$. We obtain
  $f(\ve a_1 \upto \ve a_m) = 0$, and this also holds if we
  replace~$f$ by its coefficient-wise real or imaginary part. So
  $\operatorname{Re}(f)$ or $\operatorname{Im}(f)$ (whichever is
  nonzero) provides the desired polynomial for the theorem.
\end{proof}



\begin{rem} \label{rM8}%
  \begin{enumerate}[label=(\alph*)]
  \item \label{rM8A} In principle, the polynomial~$f$ from \cref{tM8}
    can be computed by means of Gr\"obner bases. However, after a
    fortnight of running the computation for $n = 3$ with the
    state-of-the-art computer algebra system
    Magma~[\citenumber{magma}] and not seeing any result, we gave up.
  \item \label{rM8B} Nevertheless, it is practical and in fact rather
    quick to use the same approach for a {\em given} configuration of
    points~$\ve a_1 \upto \ve a_m \in \RR^n$, and thus test whether
    the matrix $A$ has full rank for all~$\ve x$. We did this
    for~$n = 3$ and~$m = 2 n + 2 = 8$. Testing ten million point
    configurations, with coordinates chosen as random integers
    between~$-100$ and~$100$, we found that for only~$11$ of them
    there existed a point~$\ve x$ such that the matrix $A$ failed to
    have full rank, and among these~$11$ configurations there was only
    one where this led to a nonunique solution of~\eqref{eqBasic}: for
    the other~$10$, the points~$\ve a_i$ were contained in a spheroid
    or paraboloid but not in a hyperboloid (see \cref{tUnique}).
  \item \label{rM8C} We also ran tests for configurations
    of~$m = 2 n + 1$ points in $\RR^n$. For such a configuration, the
    variety of all $(\ve u,\ve x,\alpha) \in \RR^{2 n + 1}$
    satisfying~\eqref{eqUAX} is (almost always) finite, and its points
    can be determined. Each point defines a quadric $Q$ with
    focus~$\ve x$ such that all points~$\ve a_i$ lie on $Q$. Now it
    can be determined whether $Q$ is a hyperboloid, and if so, whether
    all~$\ve a_i$ lie on the same sheet. Only then does the
    system~\eqref{eqBasic} have a nonunique solution (again by
    \cref{tUnique}). Experimentally we found that
    \begin{itemize}
    \item for $n = 2$ and $m = 5$, out of~$2.5$ million point
      configurations $\ve a_1 \upto \ve a_5 \in \RR^2$, slightly over
      94\% yielded a unique solution for all~$\ve x$;
    \item for $n = 3$ and $m = 7$, out of~200,000 point
      configurations $\ve a_1 \upto \ve a_7 \in \RR^3$, slightly over
      96\% yielded a unique solution for all~$\ve x$;
    \item for $n = 4$ and $m = 9$, out of~50,000 point configurations
      $\ve a_1 \upto \ve a_9 \in \RR^4$, slightly over 98\% yielded a
      unique solution for all~$\ve x$.
    \end{itemize}
    While it is not surprising that ``most'' configurations of
    $2 n + 1$ points lead to a unique solution for all~$\ve x$, an
    additional investigation would be required to interpret and
    explain the percentages we found. \remend
  \end{enumerate}
  \renewcommand{\remend}{}
\end{rem}

Of course it would be nice to have an explicit condition under which
uniqueness for all~$\ve x$ can be guaranteed. Such a condition is
provided by the following result.

\begin{theorem}[An explicit condition guaranteeing uniqueness] \label{tM10}%
  Assume that in the situation of \cref{Setup} we have
  $m \ge \binom{n+2}{2}$. If the matrix
  \[
    M := \begin{pmatrix}
      1 & a_{1,1} & \cdots & a_{1,n} & a_{1,1}^2 & \cdots &
      a_{1,n}^2 & a_{1,1} a_{1,2} & \cdots & a_{1,n-1} a_{1,n} \\
      1 & a_{2,1} & \cdots & a_{2,n} & a_{2,1}^2 & \cdots &
      a_{2,n}^2 & a_{2,1} a_{2,2} & \cdots & a_{2,n-1} a_{2,n} \\
      \vdots & \vdots & & \vdots & \vdots & & \vdots & \vdots & &
      \vdots \\
      1 & a_{m,1} & \cdots & a_{m,n} & a_{m,1}^2 & \cdots &
      a_{m,n}^2 & a_{m,1} a_{m,2} & \cdots & a_{m,n-1} a_{m,n}
    \end{pmatrix} \in \RR^{m \times \binom{n+2}{2}}
  \]
  (with $\ve a_i = (a_{i,1} \upto a_{i,n})$) has
  rank~$\binom{n+2}{2}$, then for all~$\ve x \in \RR^n$ the matrix
  $A \in \RR^{m \times (n + 2)}$, defined by~\eqref{eqA}, has
  rank~$n + 2$. So~\eqref{eqTxA} gives the unique solution
  of~\eqref{eqAbsolute}, and therefore of~\eqref{eqBasic}.

  The set of all $(\ve a_1 \upto \ve a_m)$ satisfying
  $\rank(M) = \binom{n+2}{2}$ is dense in $\RR^{m n}$.
\end{theorem}

\begin{proof}
  The homogeneous linear system given by the matrix $M$ is a system of
  equations for determining coefficients of an equation of degree
  $\le 2$ satisfied by all~$\ve a_i$. So the hypothesis means that the
  only such equation is zero, so there is no quadric containing
  $\ve a_1 \upto \ve a_m$. But then the~$\ve a_i$ certainly do not lie
  on a the same sheet of a quadric with a focus, and the claim follows
  from \cref{tUnique}\ref{tUniqueA}.

  The density statement can be seen by viewing the~$a_{i,j}$ as
  $m \cdot n$ indeterminates and realizing that the determinant of the
  the upper $\binom{n+2}{2} \times \binom{n+2}{2}$-part of $M$ is
  nonzero as a polynomial in these indeterminates.
\end{proof}

\cref{tM10} is not applicable to the GPS situation, since its
hypothesis means that the~$\ve a_i$ are not contained in a common
quadric. But the GPS satellites do fly on a common sphere (see
\url{https://www.faa.gov/about/office_org/headquarters_offices/ato/service_units/techops/navservices/gnss/gps/spacesegments}). However,
we can turn this very fact to our advantage.

\begin{theorem} \label{tM9}%
  In the situation of \cref{Setup}, assume that $m \ge \binom{n+2}{2}
  - 1$ and that the~$\ve a_i$ are contained in a common sphere (or
  another quadric that is {\em not} a hyperboloid with a focus). If
  the matrix $M \in \RR^{m \times \binom{n+2}{2}}$ from \cref{tM10}
  has rank $\binom{n+2}{2}-1$ (which is the maximal possible),
  then~\eqref{eqBasic} has a unique solution for all~$\ve x \in
  \RR^n$.
\end{theorem}

\begin{proof}
  By hypothesis, the solution space of the linear system given by $M$
  has dimension~$1$, which means that there exists only one quadric
  containing the~$\ve a_i$. Again by hypothesis, this quadric is not a
  hyperboloid with a focus, so \cref{tUnique} guarantees a unique
  solution of~\eqref{eqBasic}.
\end{proof}

\section{Conclusion} \label{sConclusion}
We have presented a comprehensive analysis of the global positioning problem in any dimension $n$ and for any number $m$ of satellites assuming that the data is exact (no noise) and the satellites are not in a common affine hyperplane. This is equivalent to the multilateration problem in which a source at an unknown location emits a signal received by $m$ sensors in known locations equipped with a synchronized clock.
Our analysis is based on a simple algebraic formulation of the problem that consists in either one linear equation \eqref{eqMatrix} or one quadratic equation in one variable together with one linear equation \eqref{eqRank4}. This formulation allowed us to fully characterize the cases where the solution of the problem is not unique and describe these cases in simple geometric terms. Our geometric description is based on quadrics of revolution ``with a focus" (see section [ref]): the satellites lie on the quadric, and  the foci of the quadric are the possible solutions for the user position. This allows us to quantify the number of cases where the solution is not unique for different numbers of satellites.
 \cref{tConclusion} summarizes some of our findings. 
Two of our results (\cref{tM5} and \cref{tM8}) required the use of Chevalley's theorem on fiber dimension, which is a staple from algebraic geometry, in a crucial step of their proofs.
 

 The case where the satellites are coplanar is still open and would be an interesting one to study. We also did not focus on the case where the user is located on the surface of the Earth and the satellites are in orbit around it. This would of be another compelling question to investigate. It would also be worthwhile to explore if our new problem formulation could be used to formulate better iterative solution methods for the case where the time measurements are corrupted by noise and possibly include outliers or extraneous data points. 

\begin{table}[htbp]
  \renewcommand{\r}[4]{\parbox[c]{20mm}{\centering
      \rule{0mm}{4mm}#1\rule[-2mm]{0mm}{0mm}} &
    \parbox[c]{20mm}{\centering
      \rule{0mm}{4mm}#2\rule[-2mm]{0mm}{0mm}} &
    \parbox[c]{40mm}{\centering
      \rule{0mm}{4mm}#3\rule[-2mm]{0mm}{0mm}} &
    \parbox[c]{20mm}{\centering
      \rule{0mm}{4mm}#4\rule[-2mm]{0mm}{0mm}} \\ \hline%
  }%
  \addtolength{\tabcolsep}{2mm}
  \begin{center}
    \begin{tabular}{|c|c|c|c|}
      \hline%
      \r{\bf Condition on~$m$}{\bf meaning for~$n = 3$}{\bf
      result}{\bf reference} \hline%
      \r{$m \le n$}{$m \le 3$}{never have
      uniqueness}{\cref{rLateration}\ref{rLaterationB}}%
      \r{$m = n+1$}{$m = 4$}{uniqueness ``probability'' strictly between~$0$ and~$1$ (experimental result)}{\cref{exUnique,exDistribution}}%
      \r{$m \ge m + 2$}{$m \ge 5$}{uniqueness for almost all~$\ve
      x$}{\cref{tM5}}%
      \r{$m = 2 n + 1$}{$m = 7$}{for ``most'' satellite positions,
      have uniqueness for all~$\ve x$ (experimental
      result)}{\cref{rM8}\ref{rM8C}}%
      \r{$m \ge 2 n + 2$}{$m \ge 8$}{for almost all satellite positions,
      have uniqueness for all~$\ve x$}{\cref{tM8}}%
    \end{tabular}
  \end{center}
  \caption{Some results from this paper that depend on~$m$, the number
    of satellites. As always, $n$ is the dimension
    and~$\ve x$ is the user position.}
  \label{tConclusion}
\end{table}

\end{document}